\documentclass[11pt]{article}
\usepackage[english]{babel}
\usepackage{amsmath,amsthm,amsfonts,amssymb}
\usepackage{fullpage}
\usepackage{setspace}
\usepackage{hyperref}
\usepackage{graphicx}
\usepackage[vlined,boxed,commentsnumbered]{algorithm2e}

\newtheorem{thm}{Theorem}[section]

\newtheorem{lem}[thm]{Lemma}
\newtheorem{prop}[thm]{Proposition}
\theoremstyle{definition}

\theoremstyle{remark}

\numberwithin{equation}{section}

\newenvironment{proofofthm}[1]{\noindent {\em Proof of Theorem #1: }\ignorespaces}{}
\newenvironment{proofoflm}[1]{\noindent {\em Proof of Lemma #1: }\ignorespaces}{}

\newcommand{\B}{{\mathbb B}}
\newcommand{\C}{{\mathbb C}}
\newcommand{\E}{{\mathbb E}}

\newcommand{\N}{{\mathbb N}}

\newcommand{\dpc}{\mathsf{C}}

\newcommand{\calH}{\mathcal{H}}
\newcommand{\calF}{\mathcal{F}}
\newcommand{\calD}{\mathcal{D}}
\newcommand{\calI}{\mathcal{I}}

\renewcommand{\P}{{\mathbb P}}

\newcommand{\e}{{\epsilon}}

\newcommand{\topic}[1]{\noindent{\underline{#1} }}

\newcommand{\poly}{\mathrm{poly}}
\newcommand{\Var}{\operatorname{Var}}
\newcommand{\ol}{\overline}
\newcommand{\wt}{\widetilde}
\newcommand{\wh}{\widehat}

\newcommand{\SK}{\textsf{SK}}
\newcommand{\SSP}{\textsf{BOSP-KC}}
\newcommand{\BOSP}{\textsf{BOSP}}
\newcommand{\SKC}{\textsf{SK-Can}}
\newcommand{\USK}{\textsf{SK-U}}
\newcommand{\SKCRC}{\textsf{SK-CC}}
\newcommand{\SBP}{\textsf{SBP}}

\newcommand{\TPM}{\textsf{TPM}}
\newcommand{\EUM}{\textsf{EUM}}
\newcommand{\EUMMono}{\textsf{EUM-Mono}}
\newcommand{\GSK}{\textsf{GenSK}}
\newcommand{\sig}{\mathsf{Sg}}
\newcommand{\conf}{\mathsf{Cf}}
\newcommand{\confl}{\mathsf{Ar}}
\newcommand{\DP}{\mathrm{DP}}

\newcommand{\segment}{\mathsf{seg}}
\newcommand{\size}{\mathrm{s}}
\newcommand{\dsize}{\mathcal{S}}
\newcommand{\Sol}{\mathsf{SOL}}
\newcommand{\CF}{\mathcal{C}}

\renewcommand{\d}{\mathrm{d}}
\newcommand{\A}{\mathfrak{A}}
\newcommand{\Pois}{\mathrm{Pois}}

\newcommand{\opt}{\mathsf{OPT}}
\newcommand{\hsigma}{\widehat{\sigma}}
\newcommand{\tchi}{\widetilde{\chi}}

\newcommand{\multi}{$\mathsf{Multi}$}
\newcommand{\exactA}{$\mathsf{Exact}$-$\mathfrak{A}$}
\newcommand{\multiA}{$\mathsf{Multi}$-$\mathfrak{A}$}
\newcommand{\MC}{$\mathsf{MC}$}
\newcommand{\MB}{$\mathsf{MB}$}
\newcommand{\MI}{$\mathsf{MI}$}

\newcommand{\eat}[1]{}

\title{{ 
\LARGE Stochastic Combinatorial Optimization via Poisson Approximation}
\thanks{Institute for Interdisciplinary Information Sciences,
Tsinghua University, China.} %
}
\author{ 
Jian Li~\footnotemark[2] \quad\quad\quad\quad Wen Yuan
\thanks{
Email: \{lijian83@mail, yuan-w10@mails\}.tsinghua.edu.cn
}
}
\date{}

\begin{document}

\pagenumbering{gobble}
\begin{titlepage}

\maketitle
\vspace{-1.2cm}
{\tiny
\nonumber
\begin{abstract}
We study several stochastic combinatorial problems, including
the expected utility maximization problem, the stochastic knapsack problem
and the stochastic bin packing problem.
A common technical challenge in these problems
is to optimize some function (other than the expectation)
of the sum of a set of random variables.
The difficulty is mainly due to the fact that the probability distribution of the sum
is the convolution of a set of distributions, which is not an easy objective function
to work with.
To tackle this difficulty, we introduce the Poisson approximation technique.
The technique is based on the Poisson approximation theorem
discovered by Le Cam, which enables us to approximate
the distribution of the sum of a set of random variables using a compound Poisson distribution.
Using the technique, we can reduce a variety of stochastic problems
to the corresponding deterministic multiple-objective problems, which either can be solved
by standard dynamic programming or have known solutions
in the literature.
For the problems mentioned above, we obtain the following results:
\begin{enumerate}
\item
We first study the expected utility maximization problem introduced recently [Li and Despande, FOCS11].
For monotone and Lipschitz utility functions, we obtain an additive PTAS if
there is a multidimensional PTAS for the multi-objective version of the problem,
strictly generalizing the previous result.
The result implies the first additive PTAS for maximizing threshold probability
for the stochastic versions of global min-cut, matroid base and matroid intersection.
\item
For the stochastic bin packing problem (introduced in [Kleinberg, Rabani and Tardos, STOC97]),
we show there is a polynomial time algorithm which uses at most the optimal number of bins,
if we relax the bin size and the overflow probability by $\e$ for any constant $\e>0$.
Based on this result, we obtain a 3-approximation if only the bin size can be relaxed by $\e$,
improving the known $O(1/\e)$ factor for constant overflow probability.
\item
For the stochastic knapsack problem, we show a $(1+\e)$-approximation using $\e$ extra capacity
for any $\e>0$, even when the size and reward of each item may be correlated
and cancelations of items are allowed.
This generalizes the previous work [Balghat, Goel and Khanna, SODA11]
for the case without correlation and cancelation.
Our algorithm is also simpler.
We also present a factor $2+\e$ approximation algorithm for stochastic knapsack with cancelations,
for any constant $\e>0$, improving
the current known approximation factor of $8$ [Gupta, Krishnaswamy, Molinaro and Ravi, FOCS11].
\item
We also study an interesting variant of the stochas- tic knapsack problem, where the size and the profit of each item are revealed before the decision is made. The problem falls into the framework of Bayesian on- line selection problems, which has been studied a lot recently.
We obtain in polynomial time a $(1+\e)$-approximate policy using $\e$ extra capacity
for any constant $\e>0$.
\end{enumerate}
Lastly, we remark that the Poisson approximation technique is quite easy to apply and may
find other applications in stochastic combinatorial optimization.

\end{abstract}

}
\end{titlepage}

\newpage

\pagenumbering{arabic}
\setcounter{page}{1}


\section{Introduction}

We study several stochastic combinatorial optimization problems, including
the threshold probability maximization problem~\cite{nikolova2006stochastic,nikolova2010approximation,li2011maximizing},
the expected utility maximization problem~\cite{li2011maximizing},
the stochastic knapsack problem~\cite{dean2008approximating,bhalgat10,gupta2011approximation},
the stochastic bin packing problem~\cite{kleinberg1997allocating, goel1999stochastic}
and some of their variants.
All of these problems are known to be \#P-hard and we are interested in
obtaining approximation algorithms with provable performance guarantees.
We observe a common technical challenge in solving these problems, that
is, roughly speaking, given a set of random variables with possibly different probability distributions,
to find a subset of random variables such that certain functional (other than the expectation
\footnote
{
We can use the linearity of expectation to circumvent the difficulty of convolution.
}
)
of their sum is optimized.
The difficulty is mainly due to the fact that the probability distribution of the sum
is the convolution of the distributions of individual random variables.
To address this issue, a number of techniques have been proposed (briefly reviewed in the related work section).
In this paper, we introduce a new technique, called the {\em Poisson approximation technique},
which can be used to approximate the probability distribution of a sum of several random variables.
The technique is very easy to use and yields better or more general results than the previous techniques
for a variety of stochastic combinatorial optimization problems mentioned above.
In the rest of the section, we formally introduce these problems and state our results.

\topic{{\bf Terminology:}}
We first set up some notations and review some standard terminologies.
Following the literature, the {\em exact version} of a problem $\A$ (denoted as \exactA) asks for
a feasible solution of $\A$ with weight exactly equal to a given number $K$.
An algorithm runs in {\em pseudopolynomial time} for \exactA\
if the running time of the algorithm is bounded by a polynomial of $n$ and $K$.

A {\em polynomial time approximation scheme (PTAS)} is
an algorithm which takes an instance of a maximization problem
and a parameter $\epsilon$ and produces a solution whose cost is at least $(1 - \epsilon)\opt$,
and the running time, for any fixed $\epsilon$, is polynomial in the size of the input.
If $\epsilon$ appears as an additive factor in the above definition, namely
the cost of the solution is at least $\opt-\e$, 
we say the algorithm
is an {\em additive PTAS}. We say a PTAS is a {\em fully polynomial time approximation scheme (FPTAS)} if the running time
is polynomial in the size of the input and $\frac{1}{\epsilon}$.

In a {\em multidimensional minimization} problem,
each element $b$ is associated with a weight vector $(w_{b,1}, w_{b,2},\ldots)$.
We are also given a budget vector $(B_1, B_2,\ldots)$.
The goal is to find a feasible solution $S\in \calF$ such that
$\sum_{b\in S} w_{b,i}\leq B_i\, \forall i$.
We use \multiA\ to denote the problem if the corresponding single dimensional optimization problem is $\A$.
A {\em multidimensional PTAS} for \multi\ is an algorithm which either
returns a feasible solution $S\in \calF$ such that $\sum_{b\in S} w_{b,i}\leq (1+\epsilon)B_i\, \forall i$,
or asserts that there is no feasible solution $S'$ with $\sum_{b\in S'} w_{b,i}\leq B_i\, \forall i$.

\subsection{Expected Utility Maximization}

We first consider the fixed set model of a class of stochastic optimization problems
introduced in \cite{li2011maximizing}.
We are given a ground set of elements (or items) $\B=\{b_{i}\}_{i=1...n}$.
Each feasible solution to the problem
is a subset of the elements satisfying some property.
In the deterministic version of the problem, we want to find a feasible solution $S$
with the minimum total weight.
Many combinatorial problems such as shortest path,
minimum spanning tree, and minimum weight matching
belong to this class.
In the stochastic version, each element $b$ is associated with a random weight $X_{b}$.
We assume all $X_{b}$s are discrete nonnegative random variables and
are independent of each other.
We are also given a utility function $\mu:\mathbb{R}\rightarrow \mathbb{R}^{+}$
to capture different {\em risk-averse} or {\em risk-prone}  behaviors
that are commonly observed in decision-making under uncertainty.
Our goal is to to find a feasible set $S$
such that the expected utility $\E[\mu(X(S))]$ is maximized,
where $X(S)=\sum_{b\in S}X_b$.
We refer to this problem as the {\em expected utility maximization (\EUM)} problem.
An important special case is to find a feasible set $S$
such that $\Pr[X(S)\leq 1]$ is maximized,
which we call the {\em threshold probability maximization (\TPM)} problem.
Note that if
$
\mu(x)= 1 \text{ for }  x\in [0,1]; \quad		\mu(x)=0 \text{ for } x>1
$,
we have that $\Pr[X(S)\leq 1]=\E[\mu(X(S))]$.
In fact, this special case has been studied
extensive in literature for various combinatorial problems
including stochastic versions of shortest path \cite{nikolova2006stochastic},
minimum spanning tree~\cite{ishii1981stochastic,geetha1993stochastic},
knapsack~\cite{goel1999stochastic} as well as some other problems~\cite{agrawal2008stochastic,nikolova2010approximation}.
We use $\A$ to denote the deterministic version of the optimization problem under consideration,
and use accordingly \EUM-$\A$ and \TPM-$\A$ to denote the expected utility maximization problem and
the threshold probability maximization problem for $\A$ respectively.

\vspace{0.3cm}
\topic{{\bf Our Results:}}
Following the previous work~\cite{li2011maximizing}, we assume $\lim_{x\rightarrow \infty}\mu(x) =0$
(if the weight of our solution is too large, it is almost useless).
We also assume $\mu$ is $\alpha$-Lipschitz in $[0,\infty)$, i.e., $|\mu(x)-\mu(y)|\leq \alpha|x-y|$
for any $x,y\in [0,\infty)$, where $\alpha$ is a positive constant.
Our first result is an alternative proof for the main result in \cite{li2011maximizing}.

\begin{thm}
\label{thm:main1}
Assume there is a pseudopolynomial time algorithm  for \textsf{Exact}-$\mathfrak{A}$.
For any $\e>0$, there is a poly-time approximation algorithm for
\EUM-$\A$ that finds a feasible solution $S$
such that
$$
\E[\mu(X(S))]\geq \opt - \e.
$$
\end{thm}

For many combinatorial problems, including shortest path, spanning tree, matching and knapsack,
a pseudopolynomial algorithm for the exact version is known.
Therefore, Theorem~\ref{thm:main1} immediately implies an additive PTAS for
the \EUM\ version of each of these problems.
An important corollary of the above theorem is a {\em relaxed} additive PTAS for \TPM:
For any $\e>0$, we can find in polynomial time a feasible solution $S$
such that
$$\Pr[X(S)\leq 1+\e]\geq \opt -\e\ ,$$
provided that there is a pseudopolynomial time algorithm  for \exactA.
In fact, the corollary follows easily by considering the monotone utility function
$
\tchi(x)= 1 \text{ for }  x\in [0,1]; \quad \tchi(x)= -\frac{x}{\e}+\frac{1}{\e}+1  \text{ for } x\in[1, 1+\e];\quad
		 \tchi(x)=0 \text{ for } x>1+\e,
$
which is $\frac{1}{\e}$-Lipschitz. We refer the interested reader to \cite{li2011maximizing} for more implications
of Theorem~\ref{thm:main1}.
However, this is not the end of story.
Our second major result considers \EUM\ with monotone nonincreasing utility functions,
a natural class of utility functions
(we denoted the problem as \EUMMono). We can get the following strictly more general result
for \EUMMono.

\begin{thm}
\label{thm:main2}
Assume there is a multidimensional PTAS for \textsf{Multi}-$\mathfrak{A}$.
For any $\e>0$, there is a poly-time approximation algorithm for \EUMMono-$\A$ that finds a feasible solution $S$
such that
$$
\E[\mu(X(S))]\geq \opt - \e.
$$
\end{thm}

It is worthwhile mentioning that condition of Theorem~\ref{thm:main2} is strictly more general than
the condition of Theorem~\ref{thm:main1}.
It is known that if there is pseudopolynomial time algorithm for \exactA,
there is a multidimensional PTAS for \multiA, by Papadimitriou and Yannakakis~\cite{papadimitriou2000approximability}.
However, the converse is not true.
Consider the minimum cut (\MC) problem. A pseudopolynomial time algorithm for $\mathsf{Exact}$-\MC\
would imply a polynomial time algorithm for the NP-hard MAX-CUT problem, while
a multidimensional PTAS for $\mathsf{Multi}$-\MC\ is known~\cite{ArmonZ06}.
Therefore, Theorem~\ref{thm:main2} implies the first relaxed additive PTAS for \TPM-\MC.
Other problems that can justify the superiority of Theorem~\ref{thm:main2} include
the matroid base (\MB) problem and the matroid intersection (\MI) problem.
Obtaining pseudopolynomial time exact algorithms for $\mathsf{Exact}$-\MB\ and $\mathsf{Exact}$-\MI\
is still open~\cite{camerini1992random}\footnote{
Pseudopolynomial time algorithms are known only for some special cases, such as spanning trees~\cite{barahona1987exact},
matroids with parity conditions \cite{camerini1992random}.
}, while multidimensional PTASes for $\mathsf{Multi}$-\MB\ and $\mathsf{Multi}$-\MI\ are known \cite{ravi1996constrained, chekuri2011multi}.

We would like to remark that obtaining an additive PTAS for \EUM-$\A$ for non-monotone utility functions
under the same condition as Theorem~\ref{thm:main2} is impossible.
Consider again \EUM-\MC.
Suppose the weights are deterministic.
The given utility function is
$
\mu(x)= 100x-89 \text{ for }  x\in [0.89,0.9]; \quad \mu(x)= -100x+91  \text{ for } x\in[0.9,0.91];\quad
		 \mu(x)=0 \text{ otherwise},
$ (which is 100-Lipschitz) and the maximum cut of the given instance has a weight $0.9$.
So, the optimal utility is $1$, but obtaining a utility value better than 0 is equivalent to finding a cut
of weight at least $0.89$, which is impossible given the imapproximability result for MAX-CUT~\cite{khot2004optimal}.

\vspace{0.2cm}
\noindent
{\bf Our techniques:}
Our algorithm consists of two major steps, discretization and enumeration.
Our discretization is similar to, yet much simpler than, the one developed by Bhalgat et al.~\cite{bhalgat10}.
In their work, they developed a technique which can discretize all (size) probability distributions
into $O(\log n)$ equivalent classes. This is a difficult task and their technique applies
several clever tricks and is quite involved.
However, we only need to discretize the distributions so that the size of the support of each distribution is a constant,
which is sufficient for the enumeration step.
In the enumeration step,
we distinguish the items with large expected weights ({\em heavy items})
and those with small expected weights ({\em light items}).
We argue that there are very few heavy items in the optimal solution, so we can afford to enumerate all possibilities.
To deal with light items, we invoke Le Cam's Poisson approximation theorem which (roughly) states that
the distribution of the total size of the set of light items can be approximated by a compound Poisson distribution,
which can be specified by the sum of the (discretized) distribution vectors of the light items
(called the {\em signature} of the set).
Therefore, instead of enumerating $O(2^n)$ combinations of light items,
we only need to enumerate all possible signatures and check whether there is a set of light items with
the sum of their distribution vectors approximately equal to (or at most)  the signature.
To solve the later task, we need the pseudopolynomial time algorithm for \exactA\
(or the multidimensional PTAS for \multiA).

\subsection{Stochastic Bin Packing}

In the stochastic bin packing (\SBP) problem, we are given
a set of items $\B = \{b_1, b_2, \ldots, b_n\}$ and an overflow probability $0<p<1$.
The size of each item $b_i$ is an independent random variable $X_i$ following a known discrete distribution.
The distributions for different items may be different.
Each bin has a capacity of $\C$.
The goal is to pack all the items in $\B$ using as few bins as possible
such that the overflow probability for each bin is at most $p$.
The problem was first studied by Kleinberg, Rabani and Tardos \cite{kleinberg1997allocating}.
They obtained a $O(1/\e)$-approximation, for only Bernoulli distributed items, if
we relax the bin size to $(1+\e)\C$ or the overflow probability to $p^{1-\e}$.
They also obtained a $O\bigl(\sqrt{\frac{\log p^{-1}}{\log\log p^{-1}}}\,\bigr)$-approximation 
without relaxing the bin size and the overflow probability. 
Goel and Indyk~\cite{goel1999stochastic} obtained PTAS for both
Poisson and exponential distributions and QPTAS (i.e., quasi-polynomial time)
for Bernoulli distribution.

\vspace{0.3cm}
\topic{{\bf Our Results:}}
Our main result for \SBP\ is the following theorem.
\begin{thm}
\label{thm:binpacking}
For any fixed constant $\e > 0$, there is a polynomial time algorithm for \SBP\ that uses at most the optimal number
of bins, when the bin size is relaxed to $(1+\e)\C$ and the overflow probability is relaxed to $p+\e$.
\end{thm}

To the best of our knowledge, our result is the first result for \SBP\ for arbitrary discrete distributions.
Based on this result,we can get the following result when the overflow probability is not relaxed.
For Bernoulli distributions, this improves the $O(1/\e)$-approximation in \cite{kleinberg1997allocating} for any constant $p$.

\begin{thm}
\label{thm:binpackingnorelax}
For any constant $p>0$, we can find in polynomial time a packing that uses at most $3\opt$ bins of capacity $(1+\e)\C$ 
such that the overflow probability of each bin is at most $p$.
\end{thm}

\vspace{0.3cm}
\topic{{\bf Our technique:}}
Our algorithm for \SBP\ is similar to that for \EUM.
We distinguish the heavy items and the light items and use the Poisson approximation technique
to deal with the light items.
One key difference from \EUM\ is that we have a linear number of bins, each of them
may hold a constant number of heavy items.
Therefore, we can not simply enumerate all configurations of the heavy items since there
are exponential many of them.
To reduce the number of the configurations to a polynomial,
we classify the heavy items into a constant number of types (again, by discretization).
For a fixed configuration of the heavy items,
using the Poisson approximation, we reduce \SBP\ to a multidimensional
version of the multi-processors scheduling problem, called the vector scheduling problem,
for which a PTAS is known \cite{chekuri1999multi}.

\subsection{Stochastic Knapsack}

The deterministic knapsack problem is a classical and fundamental problem in combinatorial optimization.
In this problem, we are given as input a set of items each associated with a size and a profit,
and our objective is to find a maximum profit subset of items whose total size is at most the capacity $\C$ of the knapsack.
In many applications, the size and/or the profit of an item may not be fixed values
and only their probability distributions are known to us in advance.
The actual size and profit of an item are revealed to us as soon as it is inserted into the knapsack.
For example, suppose we want to schedule a subset of $n$ jobs on a single machine by a fixed deadline
and the precise processing time and profit of a job are only revealed until it is completed.
In the following, we use terms {\em items} and {\em jobs} interchangeably.
If the insertion of an item causes the knapsack to overflow, we terminate and do not gain
the profit of that item.
The problem is broadly referred to as the {\em stochastic knapsack (\SK)} 
problem~\cite{dean2008approximating, bhalgat10, gupta2011approximation}.
Unlike the deterministic knapsack problem for which a solution is a subset of items,
a solution to \SK\
is specified by an {\em adaptive policy} which determines which item to insert next
based on the remaining capacity and
and the set of available items.
In contrast, a {\em non-adaptive} policy specifies a fixed permutation of items.

A significant generalization of the problem, introduced in \cite{gupta2011approximation}, considers the scenarios
where the profit of a job can be correlated with its size and we can cancel a job during its execution in the policy.
No profit is gathered from a canceled job.
This generalization is referred as {\em Stochastic Knapsack with Correlated Rewards and Cancelations (\SKCRC)}.
Stochastic knapsack and several of its variants have been studied extensively by the operation research community
(see e.g., \cite{derman1978renewal,derman1979renewal, assaf1982optimal,assaf1982renewal}).
In recent years, the problem has also attracted a lot of attention from the theoretical computer science community
where researchers study the problems from the perspective of  approximation
algorithms \cite{dean2008approximating, bhalgat10, gupta2011approximation}.

\vspace{0.3cm}
\topic{{\bf Our Results:}}
For \SK, Bhalgat, Goel and Khanna~\cite{bhalgat10} obtained a $(1+\e)$-approximation using $\e$ extra capacity.
We obtain an alternative proof of this result, using the Poisson approximation technique.
The running time of our algorithm
is $n^{O(f(\e))}$ where $f(\e)=2^{\poly(1/\e)}$, improving upon the $n^{O(f(\e))^{O(f(\e))}}$ running time in ~\cite{bhalgat10}.
Our algorithm is also considerably simpler. 

\begin{thm}
\label{thm:sk}
For any $\e>0$, there is a polynomial time algorithm that finds a $(1+\e)$-approximate
adaptive policy for \SK\ when the capacity is relaxed to $(1+\e)\C$.
\end{thm}

Our next main result is a generalization of Theorem~\ref{thm:sk} to \SKCRC\
where the size and profit of an item may be correlated, and cancelation of item in the middle is allowed.
The current best known result for \SKCRC\ is a factor 8 approximation algorithm
by Gupta, Krishnaswamy, Molinaro and Ravi~\cite{gupta2011approximation},
base on a new time-indexed LP relaxation.
We remark that it is not clear how to extend the enumeration technique developed in \cite{bhalgat10} to handle
cancelations (see a detailed discussion in Section~\ref{sec:SKCRC}).

\begin{thm}
\label{thm:skcc}
For any $\e>0$, there is a polynomial time algorithm that finds a $(1+\e)$-approximate
adaptive policy for \SKCRC\ when the capacity is relaxed to $(1+\e)\C$.
\end{thm}

We use \SKC\ to denote the stochastic knapsack problem where cancelations are allowed
(the size and profit of an item are not correlated).
Based on Theorem~\ref{thm:skcc} and the algorithm in~\cite{bhalgat20112},
we obtain a generalization of the result in \cite{bhalgat20112} as follows.

\begin{thm}
\label{thm:norelax1}
For any $\e>0$, there is a polynomial time algorithm that finds a $(2+\e)$-approximate
adaptive policy for \SKC.
\end{thm}

\vspace{0.2cm}
\topic{{\bf Bayesian Online Selection:}}
The technique developed for \SKCRC\ can be used to obtain
the following result for an interesting variant of \SK,
where the size and the profit of an item are revealed before the decision whether to select the item is made.
We call this problem the {\em  Bayesian online selection problem subject to a knapsack constraint}  (\SSP).
The problem falls into the framework of Bayesian online selection problems (\BOSP) formulated in \cite{kleinberg2012matroid}.
\BOSP\ problems subject to various constraints have attracted a lot of attention due to their applications to mechanism design 
\cite{hajiaghayi2007automated, chawla2010multi, alaei2011bayesian, kleinberg2012matroid}.
\SSP\ also has a close relation with  
the {\em knapsack secretary} problem~\cite{babaioff2007knapsack}
(See Section~\ref{sec:ssp} for a discussion).

\begin{thm}
\label{thm:ssp}
For any $\e>0$, there is a polynomial time algorithm that finds a $(1+\e)$-approximate
adaptive policy for \SSP\ when the capacity is relaxed to $1+\e$.
\end{thm}

As a byproduct of our discretization procedure, we also give a linear time FPTAS for the stochastic knapsack problem
where each item has unlimited number of copies (denoted as \USK), if we relax the knapsack capacity by $\e$.
The problem has been studied extensively
under different names and
optimal adaptive policies
are known for several special distributions~\cite{derman1978renewal,derman1979renewal, assaf1982optimal,assaf1982renewal}.
However, no algorithmic result about general (discrete and continuous) distributions is known before.
The details can be found in Appendix~\ref{app:usk}.

\noindent
\topic{\bf Our techniques:}
\SKCRC\ and \SSP\ are more technically interesting
since their solutions are adaptive
policies, which do not necessarily have polynomial size representations.
So it is not even clear at first sight where to use the Poisson approximation technique.
As before, we first discretize the distributions.
In the second step, we attempt to enumerate all possible {\em block-adaptive policies},
a notion which is introduced in \cite{bhalgat10}.
In a block-adaptive policy, instead of inserting them items one by one,
we insert the items block by block. In terms of the decision tree of a policy,
each node in the tree corresponding to the insertion of a block of items.
A remarkable discovery in \cite{bhalgat10} is that there exists a block-adaptive policy
that approximates the optimal policy and has only $O(1)$ blocks in
the decision tree (the constant depends on $\e$) for \SK.
However, their proof does not easily generalize to \SKCRC.
We extend their result to \SKCRC\ with a essentially different proof,  
which might be of independent interest.
Fixing the topology of the decision tree of the block-adaptive policy,
we can enumerate the signatures of all blocks in polynomial time,
and check for each signature whether there exists a block-adaptive policy with the signature
using dynamic programming.
Again, in the analysis, we use the Poisson approximation theorem to
argue that two block adaptive policies with the same tree topology and signatures
behave similarly.


\subsection{Other Related Work}
Recently, stochastic combinatorial optimization problems
have drawn much attention from the theoretical computer science community.
In particular, the two-stage stochastic optimization models for many classical combinatorial problem
have been studied extensively.
We refer interested reader to \cite{swamy2006approximation} for a comprehensive survey.

There is a large body of literature on \EUM\ and \TPM,
especially for specific combinatorial problems and/or special utility functions.
Loui~\cite{loui1983optimal} showed that the \EUM\ version of the shortest path problem
reduces to the ordinary shortest path (and sometimes longest path) problem
if the utility function is linear or exponential.
For the same problem,
Nikolova, Brand and Karger~\cite{nikolova2006optimal} identified more specific utility and distribution combinations
that can be solved optimally in polynomial time.
Nikolova, Kelner, Brand and Mitzenmacher~\cite{nikolova2006stochastic} studied the \TPM\ version
of shortest path when the distributions of the edge lengths are normal, Poisson or exponential.
Nikolova~\cite{nikolova2010approximation} extended this result to an FPTAS
for any problem for normal distributions, if the deterministic version of the problem
has a polynomial time algorithm.
Many heuristics for the stochastic shortest path problems
have been proposed to deal with more general utility functions 
(see e.g., \cite{murthy1997exact,murthy1998stochastic,bard1991arc}).
However, either their running times are exponential in worst cases or there is no provable performance
guarantee for the produced solution.
The \TPM\ version of the minimum spanning tree problem has been studied in~\cite{ishii1981stochastic,geetha1993stochastic},
where polynomial time algorithms have been developed for Gaussian distributed edges.

The bin packing problem is a classical NP-hard problem.
It is well known that it is hard to approximate within a factor of $3/2 - \e$ from a reduction from the subset sum problem.
Alternatively, the problem admits an asymptotic PTAS, i.e.,
it is possible to find in polynomial time a packing using at most $(1 + \e)\opt + 1$ bins for any $\e>0$ \cite{fernandez1981bin}.
The stochastic model where all items follow the same size distribution have been studied extensively
in the literature (see, e.g., \cite{coffman1980stochastic, rhee1993line}). However, these works
require that the actual items sizes are revealed before put in the bins and their focus is to
design simple rules that achieve nearly optimal packings.

Kleinberg, Rabani and Tardos \cite{kleinberg1997allocating} first considered the fixed set version of
the stochastic knapsack problem with Bernoulli-type distributions.
Their goal is to find a set of items with maximum total profit subject to the constraint that
the overflow probability is at most a given parameter $\gamma$.
They provided a polynomial-time $O(\log 1/\gamma)$-approximation.
For exponentially distributed items, Goel and Indyk \cite{goel1999stochastic} presented a bi-criterion PTAS.
Chekuri and Khanna \cite{chekuri2000ptas} pointed out that a PTAS for Bernoulli distributed items can be obtained
using their techniques for the multiple knapsack problem.
For Gaussian distributions, Goyal and Ravi \cite{goyal2009chance} obtained a PTAS.

The adaptive stochastic knapsack problem and several of it variants
have been shown to be PSPACE-hard \cite{dean2008approximating}, which implies that it is impossible in polynomial time to
construct an optimal adaptive policy, which may be exponentially large and arbitrarily complicated.
Dean, Goemans, and Vondrak~\cite{dean2008approximating} first studied \SK\
from the perspective of approximation algorithms and gave an algorithm with an approximation factor of $3+\e$.
In fact, their algorithm produces a non-adaptive policy (a permutation of items) which implies
the {\em adaptivity gap} of the problem, the maximum ratio between the expected values achieved by
the best adaptive and non-adaptive strategies, is a constant.
Using the technique developed for the $(1+\e)$-approximation using $\e$ extra capacity,
Bhalgat, Goel and Khanna~\cite{bhalgat10} also gave
an improved $(\frac{8}{3}+\e)$-approximation without extra capacity.
Stochastic multidimensional knapsack (also called stochastic packing)
has been also studied~\cite{dean2005adaptivity,bhalgat10,bansallp}.
The stochastic knapsack problem can be formulated as an exponential-size Markov decision process (MDP).
Recently, there is a growing literature on approximating the optimal policies for exponential-size MDPs in
theoretical computer science literature (see e.g.,~\cite{dean2008approximating, guha2009multi, gupta2011approximation, kleinberg2012matroid}).

\eat{
The secretary problem is a classical online selection problem introduced by Dynkin~\cite{dynkin63}.
Recently, secretary problems enjoys a revival and many generalizations have been studied extensively
(see e.g., \cite{babaioff2007matroids, babaioff2007knapsack, im2011secretary, chakraborty2012improved, jaillet2012advances})
The prophet inequalities were proposed in the seminal work of Krengel and Sucheston~\cite{krengel1977semiamarts}
and have been studied extensively since then.
Computer science researchers have identified interesting relationships between prophet inequalities and mechanism design~\cite{hajiaghayi2007automated, alaei2011bayesian}. Motivated by the relation, nearly optimal solutions are obtained recently for
the multiple-choice prophet inequalities~\cite{alaei2011bayesian} and the matroid prophet inequality~\cite{kleinberg2012matroid}.
We note that performances in these work are measured by
comparing the solutions of the online policies with the offline optimum. 
Complementarily, our work compares our policies with the optimal online policies.
}

\BOSP\ problems are often associated with the name {\em prophet inequalities}
since the solutions of the online algorithms are often compared with 
``the prophet's solutions" (i.e., the offline optimum).
The prophet inequalities were proposed in the seminal work of Krengel and Sucheston~\cite{krengel1977semiamarts}
and have been studied extensively since then.
The secretary problem is a also classical online selection problem 
introduced by Dynkin~\cite{dynkin63}.
Recently, both problems enjoy a revival due to their connections to mechanism design
and many generalizations have been studied extensively
\cite{babaioff2007matroids, babaioff2007knapsack, im2011secretary, chakraborty2012improved, jaillet2012advances,
hajiaghayi2007automated, chawla2010multi, alaei2011bayesian,  kleinberg2012matroid}.
We note that performances in all the work mentioned above are measured by
comparing the solutions of the online policies with the offline optimum. 
Complementarily, our work compares our policies with the optimal online policies.

Finally, we would like to point out that Daskalakis and Papadimitrious~\cite{daskalakis2007computing}
recently used Poisson approximation
in approximating mixed Nash Equilibria in anonymous games. However, the problem and the technique developed there
are very different from this paper. 

\vspace{0.25cm}
\topic{\bf Prior Techniques:}
As mentioned in the introduction, it is a common challenge to deal with
the convolution of a set of random variables (directly or indirectly).
To address this issue, a number of techniques have been developed in the literature.
Most of them only work for
special distributions~\cite{loui1983optimal,nikolova2006optimal,nikolova2006stochastic,goel1999stochastic, goyal2010ptas, nikolova2010approximation},
such as Gaussian, exponential, Poisson and so on.
There are much fewer techniques that work for general distributions.
Among those, the effective bandwidth technique \cite{kleinberg1997allocating} and
the linear programming technique~\cite{dean2008approximating,bansallp,gupta2011approximation}
have proven to be quite powerful for many problems,
but the approximation factors obtained are constants at best (no exception is known so far).
In order to obtain (multiplicative/additive) PTAS, two techniques are developed very recently:
one is the discretization technique~\cite{bhalgat10} for the stochastic knapsack problem
and the other is the Fourier decomposition technique~\cite{li2011maximizing} for the utility maximization problem.
However, both of them have certain limitations.
The discretization technique~\cite{bhalgat10} typically reduces a stochastic optimization problem to
a complicated enumeration problem (in some sense, it is an $O(\log n)$-dimensional optimization problem
since the distributions are discretized into $O(\log n)$ equivalent classes).
If the structure of the problem is different from or has more constraints than the knapsack problem,
the enumeration problem can become overly complicated or even intractable 
(for example, the \SKCRC\ problem or the \TPM\ version of the shortest path problem). 
In the Fourier decomposition technique~\cite{li2011maximizing},
due to the presence of complex numbers, we lose certain monotonicity property in the reduction from the stochastic optimization problem to an deterministic optimization problem,
thus it is impossible to obtain something like Theorem~\ref{thm:main2} using that technique.

\section{Expected Utility Maximization}
\label{sec:fixset}

\eat{
In this section, we do not consider correlations between item size and profit. Each item $b$ has a fixed profit $P_b$. The objective is to find a set of items $S$ with minimum overflow probability subject to the constraint that $\sum_{b \in S}P_b > P_{obj}$.
Suppose the optimal solution is $S_{OPT}$, then for any constant $0<\lambda<1$, the approximation algorithm can find a set $S$ with in time $poly(|\B|)$ such that
$$\Pr[X(S) \leq 1+\lambda]\geq\Pr[X(S_{OPT})\leq1]-\lambda$$
}

We prove Theorem~\ref{thm:main1} and Theorem~\ref{thm:main2} in this section.
For each item $b\in \B$, we use $\pi_{b}(.)$ to denote probability distribution of the weight  $X_b$ of $b$.
We use $\calF\subset 2^\B$ to denote the set of feasible solutions.
For example, in the minimum spanning tree problem,
$\B$ is the set of edges and $\calF$ is the set of all spanning trees.
Since we are satisfied with an $\epsilon$ additive approximation,
we can assume w.l.o.g. the utility function $\mu(x)=0$ if $x\geq \C$ for some constant $\C$
(e.g., we can choose $\C$ to be a constant such that $\mu(x)<\e$ if $x\geq \C$).
The support of $\pi_b(.)$ is assumed to be a subset of $[0, \C]$.
By scaling, we can assume $\C=1$ and  $0\leq \mu(x)\leq 1$ for $x\geq 0$.
We also assume $\mu$ is $\alpha$-Lipschitz where $\alpha$ is a constant that does not depend on $\e$.
It is straightforward to extend our analysis to the case where $\alpha$ depends on $\e$.
We first consider the general \EUM\ problem and then focus on the \EUMMono\ problem
where the utility function is monotone nonincreasing.

We start by bounding the total expected size of solution $S$
if $\mu(S)$ is not negligible.
This directly translates to an upper bound of the number of items with large expected weight in $S$,
which we handle separately.
The proof is fairly standard and can be found in the appendix.

\begin{lem}
\label{lm:boundexp}
Suppose each item $b \in \B$ has a non-negative random weight $X_b$
taking values from $[0,\alpha\C)]$ for some $\alpha\geq 1$.
Then,
$\forall S \subseteq U$, $\forall \frac{1}{2} > \e > 0$,
if $\E[\mu(X(S))] \geq \e$, then
$\E\left[X(S)\right] \leq 3\alpha/\e$.
\end{lem}
\eat{
\begin{proof}
Since $X_b$s are independent, $\Var[X(S)] = \sum_{b\in S}\Var[X_b]$.
As $X_b \in [0,\C]$, we have $\Var[X_b]\leq\E[X_b^2]\leq\C\cdot\E[X_b]$.
So $\Var[X_S] \leq \sum_{b\in S}\C\cdot\E[X_b] = \C\cdot\E[X(S)]$.
Suppose for contradiction that $\E[\mu(X(S))]  \geq \e$ and $\E[X(S)] > 3/\e$.
Then, we can see that
\begin{eqnarray*}
  \Var[X(S)] &>& \Pr[X(S) < \C] \cdot \left(\E[X(S)]-\C\right)^2 \\
  &\geq& \E[\mu(X(S))] \cdot \left(\E[X(S)]-\C\right)^2 \\
   &\geq& \e\cdot\left(\E[X(S)]-\C\right)^2 \\
   &=& \e\left(\E[X(S)] - 2\C + \C^2/\E[X(S)]\right)\E[X(S)] \\
   &>& \e(3/\e - 2\C + \e\C^2/3)\E[X(S)] \\
   &>& (3-2\e\C)\E[X(S)]
   > \C\cdot\E[X(S)],
\end{eqnarray*}
which contradicts the fact that $\Var[X(S)]\leq\C\cdot\E[X(S)]$.
\end{proof}
}

Let $S^*$ denote the optimal feasible set and $\opt$ the optimal value.
If $\opt=\E[\mu(X(S^*))]\leq\e$, then any feasible solution achieves
the desired approximation guarantee since $\opt-\e\leq 0$.
Hence, we focus on the other case where $\opt>\e$.
We call an item $b$ \textit{heavy item} if $\E[X_b]>\e^{10}$.
Otherwise we call it \textit{light}.
By Lemma~\ref{lm:boundexp}, we can see that
the number of heavy items in $S^*$ is at most $\frac{3}{\e^{11}}$.

\vspace{0.3cm}
\topic{Enumerating Heavy Elements}
We  enumerate all possible set of heavy items with size at most $3/\e^{11}$.
There are at most $n^{3/\e^{11}}$ such possibilities.
Suppose we successfully guess the set of heavy items in $S^*$.
In the following parts,
we mainly consider the question that given a set $H$ of heavy items,
how to choose a set $L$ of light items such that their union
$S$ is a feasible solution, and $\E[\mu(S)]$ is close to optimal.

\vspace{0.3cm}
\topic{Dealing with Light Elements}
Unlike heavy items, there may be many light items in $S^*$, which makes
the enumeration computationally prohibitive.
Our algorithm consists of the following steps.
First, we discretize the weight distributions of all items.
After the discretization, there are only a constant number of discretized weight values
in $[0,\C]$. The discretized distribution can be thought as a vector with constant dimensions.
Then, we argue that for a set $L$ of light items (with certain conditions), the
distribution of the sum of their discretized weights behaves similarly to a single item
whose weight follows a compound Poisson distribution.
The compound Poisson distribution is
completely determined by a constant dimensional vector (which we call the signature of $L$)
which is the sum of the distribution vectors in $L$.
The argument is carried out by using the Poisson approximation theorem developed by Le Cam \cite{le1960approximation}.
Then, our task amounts to enumerating all possible signatures, and checking whether
there is a set $L$ of light items with the signature and $H\cup L$ is a feasible set in $\calF$.
Since the number of possible signatures is polynomial,
our algorithm runs in polynomial time.
Now, we present the details of our discretization method.

\subsection{Discretization}
\label{sec:discretization}
In this section, we discuss how to discretize the size distributions
for items, using parameter $\e$.
W.l.o.g., we assume the range of $X_b$ is $[0, \C]$ for all $b$.
Our discretization is similar in many parts to the one in  \cite{bhalgat10},
however, ours is much simpler.

For item $b$, we say $b$ realizes to a ``large'' size if $X_b>\e^4$.
Otherwise we say $b$ realizes to a ``small'' size.
The discretization consists of two steps.
We discretize the small size region in step 1 and
the large size region in step 2.
We use $\wt X_b$ to denote the size after discretization
and $\wt{\pi}_b$ its distribution.


\vspace{0.3cm}
\topic{\bf Step 1. Small size region}
In the small size region, $\wt X_{b}$ follows a Bernoulli distribution, taking
only values $0$ and $\e^4$. The probability values $\Pr[\wt X_{b}=0]$
and $\Pr[\wt X_{b}=\e^4]$ are set such that
$$
\E[\wt X_b \mid X_b \leq \e^4] = \E[X_b \mid X_b \leq \e^4].
$$

More formally, suppose w.l.o.g. that
there is a value $0\leq d\leq \e^4$ such that
$
\Pr[X_b \geq d \mid X_b \leq \e^4] \cdot \e^4 = \E[X_b \mid X_b \leq \e^4]
$.
We create a mapping between $X_b$ and $\wt X_b$ as follows:
$$
\wt X_b =
\left\{
  \begin{array}{ll}
    0, & \hbox{$0 \leq X_b < d$;} \\
    \e^4, & \hbox{$d \leq X_b \leq \e^4$;} \\
    X_b, & \hbox{$X_b > \e^4$.}
  \end{array}
\right.
$$
In the appendix, we discuss the case where such value $d$ does not exist.

\vspace{0.3cm}
\topic{\bf Step 2. Large size region}
If $X_b$ realizes to a large size, we simply discretize it as follows:
Let $\wt X_{b} = \lfloor\frac{X_b}{\e^5}\rfloor\e^5$ (i.e., we round a large size down to a multiple of $\e^5$).

The above two discretization steps are used throughout this paper.
We denote the set of the discretized sizes by $\dsize=\{\size_0, \size_1,\ldots, \size_{z-1}\}$ where
$\size_0 = 0, \size_1 = \e^5, \size_2 = 2\e^5, \size_3 = 3\e^5, \ldots, \size_{z-1}$.
Note that $\size_1=\e^5, \ldots, \size_{1/\e-1}=\e^4-\e^5$ are also included in $\dsize$,
even though their probability is $0$.
It is straightforward to see that $ |\dsize|= z= O(\C/\e^5)$.
This finishes the description of the discretization.

\eat{
\vspace{0.3cm}
\topic{\bf Step 3. Discretizing probability values}
Suppose $\pi_b$ is the size distribution after step 1 and 2.
Note that $\pi_b(w)$ takes nonzero values only for $w$ being a multiple of $\e^5$.
The new probability distribution for $\wt X_{b}$ is defined to be
$$
\wt{\pi}_{b}(w)=\left\lceil\pi_b(w) \cdot \frac{n^2}{\e^6} \cdot (1 - \frac{2\e}{n^2}) \right\rceil \frac{\e^6}{n^2}.
$$
The sum of the new probability values is $\sum_w \wt{\pi}_{b}(w) \leq 1 - \frac{2\e}{n^2} + \frac{\C}{\e^5} \cdot \frac{\e^6}{n^2} < 1$,
and we let $\wt X_{b}$ realize to $\C+\e$ with the rest of probability.
This finishes the description of the discretization for item $b$.

For any item $b$, $\E[\wt X_b] \leq \E[X_b] \leq (1+\e)\E[\wt X_b]$ after step 1 and step 2 of discretization.
For step 3, the total variation distance of $\wt X_b$ and $X_b$ can be bounded as follows:
\begin{eqnarray*}
  \Delta(\wt X_b, X_b) &=& \sum_w \left| \wt\pi_b(w) - \pi_b(w) \right|
   \leq \sum_w \max \left\{ \frac{2\e \pi_b(w)}{n^2}, \frac{\e^6}{n^2} \right\} \\
   &\leq& \frac{2\e}{n^2} \sum_w \pi_b(w) + \frac{\e^6}{n^2} \sum_w 1
   \leq \frac{2\e}{n^2} + \frac{\e^6}{n^2} \cdot \frac{\C}{\e^5}
   < 4\e/n^2
\end{eqnarray*}
Therefore
$
\left|\E[\wt X_b] - \E[X_b]\right| < 4\e/n^2 \cdot (\C+\e) < 5\e/n^2
$
for this step.
Hence $\E[\wt X_b] - 5\e/n^2 \leq \E[X_b] \leq (1+\e)\E[\wt X_b] + 5\e/n^2$ after all three steps of the discretization.
}

The following lemma states that for a set of items, the behavior of the sum of their discretized distributions is very close to that of their original distributions.
\begin{lem}
\label{lm:fixdiscretize}
Let $S$ be a set of items such that $\E\left[X(S)\right] \leq 3/\e$.
For any $0 \leq \beta \leq \C$, we have that
\begin{enumerate}
\item $\Pr[X(S) \leq \beta] \leq \Pr[\wt X(S) \leq \beta + \e] + O(\e)$;
\item $\Pr[\wt X(S) \leq \beta] \leq \Pr[X(S) \leq \beta + \e] + O(\e)$.
\end{enumerate}
\end{lem}

\begin{lem}
\label{lm:utility}
For any set $S$ of items such that $\E\left[X(S)\right] \leq 3/\e$,
\begin{align*}
\left|\,\E[\mu(X(S))] -\E[\mu(\wt X(S))]\, \right| & = O(\e).
\end{align*}
\end{lem}
\begin{proof}
For a set $S$, we use $P_S(x)$ and $\wt{P}_S(x)$
to denote the CDFs of $X(S)$ and $\wt X(S)$ respectively.
We first observe that
\begin{align*}
\left|\,\E[\mu(X(S))] -\E[\mu(\wt X(S))]\, \right| & = \left|\,\int_{0}^{\C} \mu(x)  \d P_S(x)- \int_{0}^{\C} \mu(x) \d \wt P_S(x)\, \right| \\
&= \left|\,\int_{0}^{\C} P_S(x)\d\mu(x) - \int_{0}^{\C} \wt{P}_S(x)\d\mu(x)  \, \right| \\
& \leq \alpha \,\left|\,\int_{0}^{\C} \bigl( P_S(x)-\wt{P}_S(x) \bigr) \d x \, \right|
\end{align*}
The second equation follows from applying integration by parts and the last
is because $\mu$ is $\alpha$-Lipschitz.
From Lemma~\ref{lm:fixdiscretize}, we can see that
$$
\left|\,\int_{0}^{\C} \bigl( P_S(x)-\wt{P}_S(x) \bigr) \d x \, \right| \leq O(\e).
$$
In fact, the above can be seen as follows:
$$
\int_{0}^{\C} \bigl( P_S(x)-\wt{P}_S(x) \bigr) \d x = \int_{0}^{\C} \bigl( P_S(x)-\wt{P}_S(x+\e) \bigr) \d x
-\int_{0}^{\e}\wt{P}_S(x) \d x +  \int_{\C}^{\C+\e}\wt{P}_S(x) \d x \leq O(\e).
$$
The proof for the other direction is similar and we omit it here.
\end{proof}

\subsection{Poisson Approximation}

For an item $b$, we define its {\em signature} to be
the vector
$$
\sig(b)=\bigl(\ol\pi_b(\size_1), \ol\pi_b(\size_2), \ol\pi_b(\size_3), \ldots, \ol\pi_b(\size_{z-1})\bigr),
$$
where $\ol\pi_b(\size) = \left\lfloor\wt\pi_b(\size)\cdot\frac{n}{\e^6}\right\rfloor\cdot\frac{\e^6}{n}$ for all
nonzero discretized size $\size \in \dsize\setminus \{0\} =\{ \size_1, \size_2, \ldots, \size_{z-1} \}$.
For a set $S$ of items, its {\em signature} is defined to be the sum of the signatures of all items in $S$, i.e.,
$$
\sig(S)=\sum_{b\in S}\sig(b).
$$
We use $\sig(S)_k$ to denote the $k$th coordinate of $\sig(S)$.
By Lemma ~\ref{lm:boundexp}, $\sum_{k=1}^{z-1}\sig(S)_k \cdot \size_k=\sum_{k=1/\e}^{z-1}\sig(S)_k \cdot \size_k \leq 3/\e$.
Thus $\sig(S)_k\leq 3/\e^5$ for all $k$.
Therefore, the number of possible signatures is bounded by $\left(3n/\e^{11}\right)^{|\dsize|-1}$, which is polynomial in $n$.

For an item $b$, we let $\ol X_b$ be the random variable that $\Pr\left[\ol X_b = \size\right] = \ol\pi_b(\size)$
for $\size = \size_1, \size_2, \ldots, \size_{z-1}$, and $\ol X_b = 0$ with the rest of the probability mass.
Similarly, we use $\ol X(S)$ to denote $\sum_{b\in S}\ol X_b$ for a set $S$ of items.

The following lemma shows that it is sufficient to enumerate all possible signatures
for the set of light items.
\begin{lem}
\label{lm:fixsetsignature}
Let $S_1, S_2$ be two sets of light items such that
$\sig(S_1)=\sig(S_2)$ and
$\E\left[\wt X(S_1)\right] \leq 3/\e, \E\left[\wt X(S_2)\right] \leq 3/\e$.
Then, the total variation distance between $X(S_1)$ and $X(S_2)$ satisfies
$$
\Delta\left(\wt X(S_1), \wt X(S_2)\right) \triangleq
\sum_\size\, \left|\,\Pr\left[\wt X(S_1) = \size\right]-\Pr\left[\wt X(S_2) = \size\right]\, \right| = O(\e).
$$
\end{lem}

The following Poisson approximation theorem
by Le Cam~\cite{le1960approximation}, rephrased in our language, is essential
for proving Lemma~\ref{lm:fixsetsignature}.
Suppose we are given a $K$-dimensional vector $V=(V_1,\ldots, V_K)$.
Let $\lambda=\sum_{i=1}^K V_i$.
we say a random variable $Y$ follows the {\em compound Poisson distribution corresponding to $V$}
if it is distributed as $Y = \sum_{j=1}^{N}Y_j$
where $N$ follows Poisson distribution with expected value $\lambda$
(denoted as $N\sim \Pois(\lambda)$\,)
and $Y_1,\ldots, Y_N$ are i.i.d. random variables with $\Pr[Y_j=0]=0$ and
$\Pr[Y_j = k] = V_k / \lambda$ for $k\in \{1,\ldots, K\}$ and $j\in \{1, \ldots, N\}$.

\begin{lem}
\label{thm:poisson}
{\em \cite{le1960approximation}}
Let $X_1, X_2, \ldots$ be independent random variables taking integer values in $\{0,1,...,K\}$, let $X = \sum X_i$.
Let $\pi_i=\Pr[X_i\ne 0]$ and $V=(V_1,\ldots, V_K)$ where $V_k=\sum_{i} \Pr[X_i=k]$.
Suppose $\lambda=\sum_i \pi_i=\sum_k V_k<\infty$.
Let $Y$ be the compound Poisson distribution corresponding to vector $V$.
Then, the total variation distance between $X$ and $Y$ can be bounded as follows:
$$\Delta \Bigl(X,  Y\Bigr)
=\sum_{k\geq 0}\Bigl|\Pr[X=k]-\Pr[Y=k]\Bigr|\leq  2\sum_i \pi_i^2.$$
\end{lem}

\begin{proofoflm}{\ref{lm:fixsetsignature}}
By definition of $\ol X_b$, we have that $\Delta\left(\ol X_b, \wt X_b\right) \leq \e/n$ for any item $b$.
Since $S_1$ and $S_2$ contains at most $n$ items, by the standard coupling argument, we have that
$$
\Delta\left(\ol X(S_1), \wt X(S_1)\right) \leq \e
\quad\quad\text{ and }\quad\quad
\Delta\left(\ol X(S_2), \wt X(S_2)\right) \leq \e.
$$

If we apply Lemma~\ref{thm:poisson} to both $\ol X(S_1)$ and $\ol X(S_2)$,
we can see they both correspond to the same compound Poisson distribution, say $Y$,
since their signatures are the same.
Moreover, since the total variation distance is a metric,
we have that
\begin{eqnarray*}
  \Delta\Bigl(\wt X(S_1), \wt X(S_2)\Bigr)
   &\leq& \Delta\Bigl(\wt X(S_1), \ol X(S_1)\Bigr)
    + \Delta\Bigl(\ol X(S_1), Y\Bigr) \\
   & & + \Delta\Bigl(Y, \ol X(S_2)\Bigr)
    + \Delta\Bigl(\ol X(S_2), \wt X(S_2)\Bigr) \\
   &\leq& \e + 2\sum_{b \in S_1}\Bigl(\Pr\Bigl[\ol X_b \neq 0\Bigr]\Bigr)^2
    + 2\sum_{b \in S_2}\Bigl(\Pr\Bigl[\ol X_b \neq 0\Bigr]\Bigr)^2 + \e \\
   &=& O(\e).
\end{eqnarray*}
The last equality holds since for any light item $b$,
$$
\Pr\Bigl[\ol X_b \neq 0\Bigr] \leq \Pr\Bigl[\wt X_b \neq 0\Bigr]
= \Pr\Bigl[\wt X_b \geq \e^4\Bigr] \leq \E\Bigl[\wt X_b\Bigr] / \e^4 \leq \e^6,
$$
and
$$
\sum_{b \in S}\Pr\Bigl[\wt X_b \neq 0\Bigr] \leq \E\Bigl[\wt X(S)\Bigr] / \e^4 \leq 3/\e^5
\quad\text{ for }\quad
S = S_1, S_2.
$$
\qed
\end{proofoflm}

\subsection{Approximation Algorithm of \EUM}
Now, everything is in place to present our approximation algorithm
and the analysis.

\IncMargin{1em}
\begin{algorithm}[H]

    Discretize the size distributions of for all light items\;
    Enumerate all possible heavy item sets $H$ with $\E\bigl[\wt X(H)\bigr] < 3/\e$\;
    \For{each such $H$}{
        Enumerate all possible signatures $\sig$\;
        \For{each such $\sig$}{
            \nlset{(a)} Try to find a set $L$ of light items such that
            $H\cup L\in \calF$ ($H\cup L$ is feasible) and $\sig(L)=\sig$\;
        }
    }
    Pick the feasible $H \cup L$ with the largest $\E\bigl[\mu\bigl(\wt X(H\cup L)\bigr)\bigr]$\;

\caption{Approximation Algorithm of \EUM}\label{algo:eum}

\end{algorithm}
\DecMargin{1em}
\vspace{0.2cm}
\eat{
\begin{enumerate}
   \item Discretize the size distributions of for all light items.
   \item Enumerate all possible heavy item set $H$ with $\E[\wt X(H)] < 3/\e$.
        For each such $H$, we do the following:
        \begin{enumerate}
        \item We enumerate all possible signatures $\sig$. For each $\sig$,
         try to find a set $L$ of light items such that $H\cup L\in \calF$ ($H\cup L$ is feasible)
         and $\sig(L)=\sig$.
        \end{enumerate}
    \item Pick the feasible $H \cup L$ with the largest $\E[\mu(\wt X(H\cup L))]$.
\end{enumerate}
}

In step (a), we can use the pseudopolynomial time algorithm for the exact version of the problem
to find a set $L$ with the signature exact equal to $\sig$.
Since $\sig$ is a vector with $O(\e^{-5})$ coordinates and the value of
each coordinate is bounded by O(n), it can be encoded by an integer which is at most $n^{O(\e^{-5})}$.
Thus the pseudopolynomial time algorithm actually runs in $\poly(n, n^{O(\e^{-5})})=n^{O(\e^{-5})}$ time, which is a polynomial.
Since there are at most $N^{O(\e^{-11})}$ different heavy item sets and
$n^{O(\e^{-5})}$ different signatures,  the algorithm runs in
$n^{O(\e^{-15})}=n^{\poly(1/\e)}$ time overall.
Finally, we present the analysis of the performance guarantee of the algorithm.

\eat{
\begin{thm}
\label{thm:main1}
Assume there is a pseudopolynomial time algorithm  for the exact version of problem $A$.
For any $\e>0$, there is a polynomial time approximation algorithm for \EUM($A$) that finds a feasible solution $S\in \calF$
such that
$$
\E[\mu(X(S))]\geq \opt - \e.
$$
\end{thm}
}

\begin{proofofthm}{\ref{thm:main1}}
Assume the optimal feasible set is $S^*=H^*\cup L^*$
where items in $H^*$ are heavy and items in $L^*$ are light.
Assume our algorithm has guessed $H^*$ correctly.
Since there is a pseudopolynomial algorithm,
we can find a set $L$ of light items such that $\sig(L)=\sig(L^*)$.
By Lemma~\ref{lm:fixsetsignature}, we know that
$
\Delta\Bigl(\wt X(L), \wt X(L^*)\Bigr) =  O(\e).
$
Therefore, we can get that
$$
\Delta\Bigl(\wt X(L\cup H^*), \wt X(L^*\cup H^*)\Bigr) =  O(\e).
$$
Moreover, we have that
\begin{align*}
\Bigl|\,\E\Bigl[\mu\bigl(\wt X(L\cup H^*)\bigr)\Bigr]  -\E\Bigl[\mu\bigl(\wt X(L^*\cup H^*)\bigr)\Bigr]\, \Bigr|
&\leq  \sum_{\size\in \dsize} \mu(\size) \, \Bigl|\wt{\pi}_{L\cup H^*}(\size)-\wt{\pi}_{L^*\cup H^*}(\size)\Bigr| \,  \\
& \leq  \max_{0\leq x\leq \C}\mu(x)\cdot \Delta\Bigl(\wt X(L\cup H^*), \wt X(L^*\cup H^*)\Bigr) =O(\e),
\end{align*}
where $\wt\pi_S$ is the PDF for $\wt X(S)$.
It is time to derive our final result:
\begin{align*}
\left|\,\E[\mu(X(L^*\cup H^*))]- \E[\mu(X(L\cup H^*))] \, \right| & \leq
 \left|\,\E[\mu(X(L^*\cup H^*))] -\E[\mu(\wt X(L^*\cup H^*))]\, \right| \\
&+ \left|\,\E[\mu(\wt X(L^*\cup H^*))-\E[\mu(\wt X(L\cup H^*))] ]\, \right| \\
&+\left|\,\E[\mu(\wt X(L\cup H^*))-\E[\mu(X(L\cup H^*))] ]\, \right| \\
&= O(\e)
\end{align*}
The inequality follows from Lemma~\ref{lm:utility} and~\ref{lm:fixsetsignature}.\qed
\end{proofofthm}

\subsection{Approximation Algorithm for \EUMMono}
We prove Theorem~\ref{thm:main2} in this subsection.
Recall that \EUMMono\ is a special case of \EUM\ where the utility function $\mu$ is monotone nonincreasing.
The algorithm is the same as that in \EUM\ except we adopt the new step (a), as follows.

\IncMargin{1em}
\begin{algorithm}[H]

    \nlset{(a)} Try to find a set $L$ of light items such that
    $H\cup L\in \calF$ ($H\cup L$ is feasible)
    and $\sig(L)\leq (1+\e^6)\sig$ (coordinatewise)\;

\end{algorithm}
\DecMargin{1em}

\begin{lem}
\label{lm:cpddomination}
We are given two vectors $V_1\leq V_2$ (coordinatewise).
$Y_1$ and $Y_2$ are random variables
following CPD corresponding to $V_1$ and $V_2$, respectively.
Then, $Y_2$ stochastically dominates $Y_1$.
\end{lem}
\begin{proof}
We are not aware of an existing proof of this intuitive fact,
so we present one here for completeness.
The lemma can be proved directly from the definition of CPD, but the proof is tedious.
Instead,  we use Lemma~\ref{thm:poisson} to give an easy proof as follows:
Consider the sum $X$ of a large number $N$ of nonnegative random variables $\{X_i\}_{i=1,\ldots, N}$, each $X_i$
having a very small expectation. Suppose $\sum_i \sig(X_i)=V_1$.
As $N$ goes to infinity and each $\E[X_i]$ goes to 0, the distribution of $X$
approaches to that of $Y_1$ since their total variation distance approaches to 0.
We can select a subset $S$ of $\{X_i\}_{i=1,\ldots, N}$ so that $\sum_{i\in S}\sig(X_i)= V_2$.
So, the sum of the subset, which approaches to $Y_2$ in the limit, is clearly stochastically dominated by total sum $X$.
\end{proof}

\begin{lem}
\label{lm:fixsetsignature1}
Let $S_1, S_2$ be two sets of light items with
$\E\bigl[\wt X(S_1)\bigr] \leq 3/\e$ and $\E\bigl[\wt X(S_2)\bigr] \leq 3/\e$.
If $\sig(S_1)\leq (1+\e^6)\sig(S_2)$,
then we have that for any $\beta > 0$
$$
\Pr[\wt X(S_1)\leq \beta] \geq \Pr[\wt X(S_2)\leq \beta]-O(\e).
$$
\end{lem}
\begin{proof}
Let $Z_1$ and $Z_2$ be the compound Poisson distribution (CPD) corresponding to
$\sig(S_1)$
and $\sig(S_2)$, respectively.
Denote $\lambda=\|\sig(S_2)\|_1$.
Let $Y$ be the CPD defined as $Y=\sum_{i=1}^N Y_i$
where $N\sim \Pois((1+\e^6)\lambda)$ and $Y_i$s are i.i.d. random variables with
$\Pr[Y_i=\size_k]=\frac{\sig(S_2)_k}{\lambda}$ for each $\size_k\in \dsize\setminus\{\size_0\}$.
By Lemma~\ref{thm:poisson},
$Z_2$ is distributed as $\sum_{i=1}^{N'} Y_i$
where $N'\sim \Pois(\lambda)$.
By the standard coupling argument, we can see that
\begin{align*}
\Delta(Y, Z_2)\leq \Delta(N,N')=\Delta(\Pois((1+\e^6)\lambda),\Pois(\lambda)) = O(\e).
\end{align*}

This is because $\Pr[N = k] = (1+\e^6)^k e^{-\e^6\lambda} \Pr[N' = k]$ for all $k\in\N$.
Since $\lambda=\|\sig(S_2)\|_1 = O(\e^{-5})$, $(1+\e^6)^k e^{-\e^6\lambda} \geq e^{-\e^6\lambda} = 1-O(\e)$.
Therefore the total variation distance of $N$ and $N'$ can be bounded by $O(\e)$.

Since $\sig(S_1)\leq (1+\e^6)\sig(S_2)$, $Z_1$ (the CPD corresponding to $\sig(S_1)$) is stochastically dominated by
$Y$ (the CPD corresponding to $(1+\e^6)\sig(S_2)$) by Lemma~\ref{lm:cpddomination}.
Therefore,
\begin{align*}
\Pr[\ol X(S_1)\leq \beta] & \geq \Pr[Z_1\leq \beta] - \Delta(\ol X(S_1), Z_1) \\
& \geq \Pr[Y\leq \beta] - O(\e) \\
& \geq \Pr[Z_2\leq \beta]-\Delta(Y, Z_2) -O(\e) \\
& \geq \Pr[\ol X(S_2)\leq \beta]-\Delta(\ol X(S_2), Z_2) -O(\e) \\
&\geq \Pr[\ol X(S_2)\leq \beta]-O(\e)
\end{align*}
We also have $\Delta\Bigl(\wt X(S), \ol X(S)\Bigr)\leq \e$ for $S = S_1, S_2$. Thus
$$
\Pr[\wt X(S_1)\leq \beta] \geq \Pr[\wt X(S_2)\leq \beta]-O(\e).
$$
This completes the proof of the lemma.
\end{proof}

\eat{
\begin{thm}
\label{thm:main2}
Assume there is a multidimensional PTAS for the multidimensional minimization version of the problem $A$.
For any $\e>0$, there is a polynomial time approximation algorithm for \EUMMono($A$) that finds a feasible solution $S\in \calF$
such that
$$
\E[\mu(X(S))]\geq \opt - \e.
$$
\end{thm}
}

\begin{proofofthm}{\ref{thm:main2}}
The proof is similar to the proof of Theorem~\ref{thm:main1}.
Assume the optimal feasible set is $S^*=H^*\cup L^*$
where items in $H^*$ are heavy and items in $L^*$ are light.
Assume our algorithm has guessed $H^*$ correctly.
Since there is a multidimensional PTAS,
we can find a set $L$ of light items such that
$\sig(L)\leq (1+\e^6)\sig(L^*)$.
By Lemma ~\ref{lm:fixsetsignature1}, we know that for any $\beta > 0$,
$\Pr[\wt X(L)\leq\beta]\geq\Pr[\wt X(L^*)\leq\beta]-O(\e)$.
Therefore, we can get that for any $\beta > 0$,
$
\Pr[\wt X(L \cup H^*)\leq \beta] \geq \Pr[\wt X(L^* \cup H^*)\leq \beta]-O(\e).
$
Now, we can bound the expected utility loss for discretized distributions:
\begin{align*}
\E[\mu(\wt X(L^*\cup H^*))]-\E[\mu(\wt X(L\cup H^*))]\,
&= \int_0^\C \Bigl(\wt{P}_{L\cup H^*}(x)-\wt{P}_{L^*\cup H^*}(x)\Bigr)  \d \mu(x) \\
&\leq \alpha \cdot \left| \int_0^\C \Bigl(\wt{P}_{L\cup H^*}(x)-\wt{P}_{L^*\cup H^*}(x)\Bigr)  \d x \right| \\
&=O(\e).
\end{align*}
Finally, we can show the performance guarantee of our algorithm:
\begin{align*}
\E[\mu(X(L^*\cup H^*))]-\E[\mu(X(L\cup H^*))]  & \leq
\left|\,\E[\mu(X(L^*\cup H^*))] -\E[\mu(\wt X(L^*\cup H^*))]\, \right| \\
&+ \left|\,\E[\mu(\wt X(L^*\cup H^*))] -\E[\mu(\wt X(L\cup H^*))]\, \right| \\
&+ \left|\,\E[\mu(\wt X(L\cup H^*))] - \E[\mu(X(L\cup H^*))]\, \right| \\
&= O(\e)
\end{align*}
The last inequality follows from Lemma~\ref{lm:utility}.\qed
\end{proofofthm}


\section{Stochastic Bin Packing}
\label{sec:binpacking}

Recall that in the stochastic bin packing (\SBP) problem, we are given
a set of items $\B = \{b_1, b_2, \ldots, b_n\}$ and an overflow probability $0<p<1$.
The size of each item $b_i$ is an independent random variable $X_i$.
The goal is to pack all the items in $\B$ into bins with capacity $\C$, using as few bins as possible,
such that the overflow probability for each bin is at most $p$.
The main goal of this section is to prove Theorem~\ref{thm:binpacking}.

W.l.o.g., we can assume that $p \leq 1-\e$ where $\e$ is the error parameter.
Otherwise, the overflow probability is relaxed to $p+\e \geq 1$,
and we can pack all items in a single bin.
Let the number of bins used in the optimal solution be $m$.
In our algorithms, we relax the bin size to $\C + O(\e)$, which is less than $2\C$. W.l.o.g., we assume the support of$ X_i$ is $[0,2\C]$.
From now on, assume that our algorithm has guessed $m$ correctly.
We use $B_1, B_2, \ldots, B_m$ to denote the bins.

\subsection{Discretization}
\label{sec:bp_conf}
We first discretize the size distributions for all items in $\B$, using parameter $\e$,
as described in Section~\ref{sec:discretization}.
Denote the discretized size of $b_i$ by $\wt X_i$.

We call item $b_i$ a {\em heavy} item if $\E[X_i] \geq \e^{15}$. Otherwise, $b_i$ is {\em light}.
We need to further discretize the size distributions of the heavy items.
We round down the probabilities of $\wt X_i$ taking each nonzero value to multiples of $\e^{22}$.
Denote the resulting random size by $\wh X_i$.
More formally, $\Pr[\wh X_i=\size] = \lfloor \Pr[\wt X_i=\size]\cdot \e^{-22} \rfloor \cdot \e^{22}$ for any
$\size \in \dsize\setminus\{0\}$.
Use $\calH$ to denote the set of all discretized distributions for heavy items.
We can see that $|\calH| = (\e^{-22})^{|\dsize|}= (\e^{-22})^{O(\e^{-5})}$.
Denote them by $\Pi_1, \Pi_2, \ldots, \Pi_{|\calH|}$ (in an arbitrary order).

For a set $S$ of items, we use $H(S)$ to denote the set of heavy items in $S$, and use $L(S)$ to denote the set of light items in $S$.
We define the {\em arrangement} for heavy items in $S$ to be the $|\calH|$-dimensional vector:
$$\confl(S) = (N_1, N_2, \ldots, N_{|\calH|})$$
where $N_k\in \N$ is the number of heavy items in $S$ following the discretized size distribution $\Pi_k$, $k=1,2,\ldots,|\calH|$.
Suppose we pack all items in $S$ into one bin.
By Lemma ~\ref{lm:boundexp} and the assumption that $p \leq 1-\e$, $\E\bigl[\wt X(S)\bigr]\leq  3/\e$.
So, we can pack at most $O(\e^{-16})$ heavy items into a bin.
Therefore, the number of possible arrangements for a bin is bounded by ${|\calH|+O(\e^{-16})\choose |\calH|}$, which is a constant.

Let the {\em signature} of a light item $b$  be
$
\sig(b)=\bigl(\wt\pi_b(\size_i)\bigr)_{1\leq i\leq |\dsize|-1}
$
(Note that the definition is slightly different from the previous one).
The {\em signature} of the a set $S$ of light items is defined to be
$
\sig(S) = \sum_{b \in S}\sig(b).
$
If $S$ consists of both heavy and light items,
we use $\sig(S)$ as a short for $\sig(L(S))$.
Moreover, for set $S$, we define the {\em rounded signature} to be
$$
\wh\sig(S)\,:\,\wh\sig(S)_k = \left\lceil\sig(S)_k\cdot\e^{-6}\right\rceil\cdot\e^6\quad\text{ for each }k=1,\ldots, |\dsize|-1.
$$
Suppose we pack all items in $S$ into one bin.
Since $\E\bigl[\wt X(S)\bigr] \leq 3/\e$, $\wh\sig(S)_k \leq O(1/\e^5)$ for any $k$.
Therefore, the number of possible rounded signatures is bounded by $(3/\e^{11})^{|\dsize|-1}$.

The {\em configuration} of a set $S$ of items is defined to be $\conf(S) = \bigl(\confl(S); \wh\sig(S)\bigr)$.
It is straightforward to see the number of all configurations is bounded by
$$
h=O\biggl(\frac{1}{\e^{11}}\biggr)^{|\dsize|-1}{|\calH|+O(\e^{-16})\choose |\calH|},
$$
which is still a constant.

We also define the {\em s-configuration} of a solution $\Sol=\{S_1,\ldots, S_m\}$ ($S_m$ is the set of items packed in bin $B_i$)
to be
$$
\conf(\Sol) = \bigl\{\conf(S_1), \conf(S_2), \ldots, \conf(S_m)\bigr\}.
$$
We note that $\conf(\Sol)$ is a multi-set (instead of a vector),
i.e., the indices of the bins do not matter.
Hence, the number of all possible s-configurations is bounded by ${{m + h} \choose m} = \poly(m)$.
Let $\CF$ be the set of all possible s-configurations.

\subsection{Our Algorithm}
\label{sec:bp_alg}

Before describing our algorithm,
we need a procedure to solve the following multi-dimensional optimization problem:
We are given an s-configuration
$$\conf = \bigl((\confl_1, \wh\sig_1), (\confl_2, \wh\sig_2), \ldots, (\confl_m, \wh\sig_m)\bigr),$$
Our goal is to find a packing $\Sol=(S_1,\ldots, S_m)$
such that
$\confl(S_i) = \confl_i$ and $\sig(S_i)\leq \wh\sig_i+\e^6\bf1$ (where ${\bf1}=(1,\ldots,1)$) for $1\leq i\leq m$
or to claim that
there is no packing $\Sol=(S_1,\ldots, S_m)$
such that
$\confl(S_i) = \confl_i$ and $\sig(S_i)\leq \wh\sig_i$ for $1\leq i\leq m$.
If we succeed in finding such a solution,
we say $\conf$ {\em passes the test of feasibility},
otherwise we say $\conf$ fails the test.

Finding a solution such that $\confl(S_i) = \conf_i$ for all $i$ is trivial.
Now, we concentrate on the set of the light items, $L(\B)$.
In fact, the problem for light items becomes a variant of the multidimensional bin packing problem, called the {\em vector scheduling} problem,
which has been studied in~\cite{chekuri1999multi}.
For completeness, we sketch their approach, using our notations.
We write a linear integer program, solve its LP relaxation and then round the solution to a feasible packing.
We use the Boolean variables $x_{ij}$ to denote whether the light item $b_i$ is packed into bin $B_j$.
We have the integrality constraints $x_{ij} \in \{0,1\}$ and they are relaxed to $x_{ij} \geq 0$ in the following LP relaxation:

\begin{enumerate}
  \item $\sum_{j=1}^m x_{i,j} = 1,
            \quad\quad i \in L(\B),$
  \item $\sum_{i\in L(\B)} \wt\pi_i(\size_k) \cdot x_{ij} \leq (\wh\sig_j)_k,
            \quad\quad j\in\{1,\ldots, m\}, \,k\in \{1,\ldots, |\dsize|-1\},$
  \item $x_{ij}\geq 0,  \quad\quad i \in L(\B), \,j\in\{1,\ldots, m\}$.
\end{enumerate}
The following proposition states a well-known property for any basic solution of the LP.
\begin{prop}
\label{prop:bp_LP}
Any {\em basic feasible solution} to the LP has at most $(|\dsize|-1)\cdot m$ light items that are packed fractionally into more than one bins.
\end{prop}

If the above LP has no feasible solutions, we say $\conf$ fails the test.
Otherwise $\conf$ passes the test, and we find a solution $\Sol$ as follows.
First, we solve the LP and obtain a basic feasible solution.
Let $F$ be the set of light items that are fractionally packed into more than bins.
By proposition ~\ref{prop:bp_LP}, $|F|\leq(|\dsize|-1)\cdot m$.
We partition $F$ arbitrarily into $m$ subsets, each containing at most $(|\dsize|-1)$ items,
and then pack the $k$-th subset into the $k$-th bin.
Since the expected size of a light item is less than $\e^{15}$, $\wt\pi_i(\size_k)\leq\e^{11}$ for any $i, k$.
Therefore, $\sig(S_j) \leq \wh\sig_j + (|\dsize|-1)\e^{11} = \wh\sig_j + \e^6\bf1$ for any $j$.

We need one more notation to describe our algorithm.
For a set of items $S$, let
$$
\Pr(\conf(S), \C) =
\Pr((\confl(S),\wh\sig(S)), \C) =
 \Pr\bigl[\wh X_H + \wh Y_L \geq \C\bigr]
$$
where $\wh X_H = \wh X\bigl(H(S)\bigr)$ and
$\wh Y_L$ is the CPD corresponding to $\wh \sig(S)$ (according to Lemma ~\ref{thm:poisson}) .
By definition, if two sets $S_1$ and $S_2$ have the same configuration, 
$\Pr(\conf(S_1), \C) = \Pr(\conf(S_2), \C)$.

Now, everything is ready to state our algorithm.
We simply enumerate all s-configurations in $\CF$.
For each s-configuration $\conf$,
we first compute $\Pr(\conf(S), \C)$.
If it is at most $p+O(\e)$, we run the feasibility test.
If $\conf$ passes the test, the returned solution is our final packing.
The pseudocode of our algorithm is described in Algorithm~\ref{algo:bp}.

\IncMargin{1em}
\begin{algorithm}[t]
    \BlankLine
    \For($\quad\quad//$ guess the number of bins in $\opt$){$m \leftarrow 1$ \KwTo $n$}{
        \For{each s-configuration $\conf=(\conf_1,\ldots,\conf_m)\in\CF$}{
            \If{$\Pr(\conf_j, (1+O(\e))\C) \leq p + O(\e), \forall j$ and $\conf$ pass the feasibility test}
            {
                    Return the solution $S$ obtained from the testing algorithm\;
            }
        }
    }
\caption{Stochastic Bin Packing}\label{algo:bp}
\end{algorithm}
\DecMargin{1em}

The algorithm clearly runs in polynomial time since the number of s-configurations is polynomial
and the feasibility test also runs in polynomial time.

\subsection{Analysis}

The following lemma shows that we can approximate the overflow probability of a bin given its configuration.
Therefore, it is sufficient to enumerate all possible configurations for each bin to find an approximation solution.

\begin{lem}
\label{lm:bp_conf}
For any set $S$ consisting of at most $3/\e^{16}$ heavy items, we have that
$$
\Bigl|\,\Pr(\conf(S), \C) - \Pr\bigl[\wt X(S) \geq \C\bigr]\,\Bigr| = O(\e).
$$
\end{lem}
\begin{proof}
Let $\wh X_H = \wh X\bigl(H(S)\bigr)$ and
$\wh Y_L$ be the CPD corresponding to $\wh \sig(S)$.
Since there are at most $3/\e^{16}$ heavy items, by the coupling argument,
$$
\Delta\Bigl(\wh  X_H, \wt X\bigl(H(S)\bigr)\Bigr) \leq 3/\e^{16} \cdot |\dsize| \cdot \e^{22} = O(\e).
$$
Let $Y_L$ be the CPD corresponding to $\sig(S)$. By Lemma ~\ref{lm:fixsetsignature},
we can see that
$$
\Delta\Bigl(Y_L, \wt X\bigl(L(S)\bigr)\Bigr) = O(\e).
$$
We can also show $\Delta\bigl(\wh Y_L, Y_L\bigr) = O(\e)$ since
$\bigl\|\,\sig(S) - \wh\sig(S)\,\bigr\|_\infty = O(\e^6)$.
Therefore,
\begin{eqnarray*}
  \Delta\Bigl(\wh X_H + \wh Y_L, \wt X(S)\Bigr)
   \leq \Delta\Bigl(\wh X_H, \wt X\bigl(H(S)\bigr)\Bigr)
    + \Delta\Bigl(\wh Y_L, Y_L\Bigr)
    + \Delta\Bigl(Y_L, \wt X\bigl(L(S)\bigr)\Bigr)
   = O(\e)
\end{eqnarray*}
This finishes the proof of the lemma.
\end{proof}

The following lemma shows that
we can approximate the overflow probability even if only an approximate signature is given.

\begin{lem}
\label{lm:bp_sig}
For any two sets $S_1$, $S_2$ such that $\confl(S_1) = \confl(S_2)$ and $\sig(S_1) \leq \sig(S_2) + 3\e^6\bf1$,
we have that
$$
\Pr\bigl[\wt X(S_1) \geq \C\bigr] \leq \Pr\bigl[\wt X(S_2) \geq \C\bigr] + O(\e).
$$
\end{lem}

\begin{proof}
Let $\sig_S = \sig(S_2) + 3\e^6\bf1$.
Let $\wh X_H= \wh X\bigl(H(S_1)\bigr)=\wh X\bigl(H(S_2)\bigr)$ (since $\confl(S_1) = \confl(S_2)$).
Let $Y_1, Y_2$ be the CPD corresponding to $\wh \sig(S_1)$ and $\wh \sig(S_2)$, respectively.
Let $Z, Z_1, Z_2$ be the CPD corresponding to $\sig_S, \sig(S_1)$ and $\sig(S_2)$, respectively.
Since $\bigl\|\,\wh \sig(S_i) - \sig(S_i)\,\bigr\|_\infty = O(\e^6)$, we have that $\Delta\bigl(Y_i, Z_i\bigr) = O(\e)$
for $i=1,2$ (The proof is almost the same as that of Lemma~\ref{lm:fixsetsignature} and omitted here).
By Lemma~\ref{lm:bp_conf} we have that
$
\bigl|\,\Pr(\conf(S_i), \C) - \Pr[\wt X(S_i) \geq \C]\,\bigr| = O(\e)
$
for $i=1,2$.
Since $\bigl\|\,\sig_S - \sig(S_2)\,\bigr\|_\infty = 3\e^6$, we have that $\Delta\bigl(Z, Z_2\bigr) = O(\e)$.
Combining these facts together, we have that
\begin{align*}
\Pr\bigl[\wt X(S_2) \geq \C\bigr] & \geq \Pr\bigl[\wh X_H + Y_2 \geq \C\bigr] - O(\e)
  \geq \Pr\bigl[\wh X_H + Z_2 \geq \C\bigr] - O(\e) \\
 & \geq \Pr\bigl[\wh X_H + Z \geq \C\bigr] - O(\e).
\end{align*}
On the other hand, $\sig(S_1) \leq \sig_S$. So, by Lemma~\ref{lm:cpddomination},
$Z_1$ is stochastically dominated by $Z$. Therefore, we have that
\begin{align*}
\Pr\bigl[\wt X(S_1) \geq \C\bigr] & \leq \Pr\bigl[\wh X_H + Y_1 \geq \C\bigr] + O(\e)
 \leq \Pr\bigl[\wh X_H + Z_1 \geq \C\bigr] + O(\e) \\
& \leq \Pr\bigl[\wh X_H + Z \geq \C\bigr] + O(\e).
\end{align*}
Combining these two results, we complete the proof.
\end{proof}

Now, we are ready to prove the main theorem of this section.
\vspace{0.3cm}

\begin{proofofthm}{~\ref{thm:binpacking}}
Suppose the optimal solution $\opt=(O_1,\ldots, O_m)$ uses $m$ bins.
The algorithm will enumerate its s-configuration $\conf(\opt)$.
Obviously, $\conf(\opt)$ can pass the feasibility test.
Let $\Sol=(S_1,\ldots,S_m)$ be the solution obtained by the LP rounding procedure,
which guarantees that $\sig(S_j)\leq \sig(O_j)+3\e^6\bf1$, for any $j$.
By Lemma~\ref{lm:bp_sig} and Lemma ~\ref{lm:fixdiscretize}, for $1\leq j\leq m$,
\begin{align*}
\Pr\bigl[X(S_j) \geq (1+O(\e))\C\bigr] & \leq
\Pr\bigl[\wt X(S_j) \geq (1+O(\e))\C\bigr]+O(\e) \leq \Pr\bigl[\wt X(O_j) \geq (1+O(\e))\C\bigr] + O(\e) \\
& \leq \Pr\bigl[X(O_j) \geq \C\bigr] + O(\e) \leq p + O(\e).
\end{align*}
The proof of the theorem is completed.
\qed
\end{proofofthm}

\subsection{A 3-Approximation without Relaxing the Overflow Probability}
\label{sec:bp_norelax}
For \SBP, we can find a 3-approximation
within polynomial time without relaxing the overflow probability $p$,
for any constant $0<p<1$.
First, we note that for any set of items,
we can estimate the overflow probability
by using the technique for counting knapsack solutions~\cite{dyer2003approximate}.
In fact, since we assume $p$ is constant, we can simply use the
Monte Carlo method to get an estimate with additive error $\e$ with high probability
by randomly taking $O(\log n)$ samples.
For each item set $S$ we use $P(S)$ to denote the estimate probability of $\Pr[X(S)\geq\C+\e]$
Then, we have
\begin{description}
  \item[(a)] $(1-\e)\Pr[X(S)\geq\C+\e] \leq P(S)\leq (1+\e)\Pr[X(S)\geq\C+\e]$
\end{description}
with high probability ($1-\frac{1}{n^c}$ for some constant $c$) when $\Pr[X(S)\geq\C+\e] \geq p$.

We first run Algorithm ~\ref{algo:bp} and obtain a packing $B_1,\ldots, B_m$
where $S_i$ is the set of items in $B_i$ and $m\leq \opt$.
We know that $\Pr[X(B_i)\geq \C+\e]\leq p+\e$.
Then, for each $B_i$, we distribute the items in $B_i$ to at most 3 new bins
such that the overflow probability is not at most $p$ for each new bin.
Let $\e$ be any constant less than $\frac{p-p^2}{1+4p}$.
The pseudo-code of our algorithm is described in Algorithm ~\ref{algo:bp_norelax}.

\IncMargin{1em}
\begin{algorithm}[t]
    Run Algorithm~\ref{algo:bp} to get a packing $S_1, S_2, \ldots, S_m$\;
    \For{ each bin $B_i$ }{
        Pack each item in $B_i$ into a new bin $B_{i,j}$\;
        \While{ there exist two bins $B_{i,j_1}$, $B_{i,j_2}$ that $P(X(B_{i,j_1}\cup B_{i,j_2})) \leq (1-\e)p$ }{
        \nlset{(a)}    Merge $B_{i,j_1}$ and $B_{i,j_2}$ into one bin\;
        }
    }
\caption{
    \vspace{0.5cm}
Stochastic Bin Packing without Relaxing the Overflow Probability}\label{algo:bp_norelax}
\end{algorithm}
\DecMargin{1em}

\begin{proofofthm}{\ref{thm:binpackingnorelax}}
Now we claim that for each $B_i$, the while loop (a) produces
at most 3 new bins.
The approximation factor of 3 follows immediately since $m\leq \opt$.

Now, we prove the claim.
For each bin $B_{i,j}$ output by the algorithm, either it packs only one item
with overflow probability at most $p$ (otherwise, there is no feasible solution),
or $P(B_{i,j}) \leq (1-\e)p$.
Therefore, the true overflow probability of $B_{i,j}$ is at most $p$.

We still need to show the while loop terminates with at most 3 bins.
Suppose for contradictor that it outputs at least 4 bins, and
$B_{i,1}, B_{i,2}, B_{i,3}, B_{i,4}$ are four of them.
Then, $P(B_{i,1} \cup B_{i,2}) > (1-\e)p$.
Therefore, $\Pr[X(B_{i,1} \cup B_{i,2}) \geq \C+\e] > (1-2\e)p$.
Similarly, $\Pr[X(B_{i,3} \cup B_{i,4}) \geq \C+\e] > (1-2\e)p$.
Thus, we have that
$$
\Pr[X(B_i) \geq \C+\e] > 1-(1-(1-2\e)p)^2 > p + \e
$$
for any $\e< \frac{p-p^2}{1+4p}$.
This contradicts to the fact that $\Pr[X(B_i) \geq \C+\e] \leq p + \e$.
\qed
\end{proofofthm} 

\section{Stochastic Knapsack}

An instance of the stochastic knapsack problem
can be specified by a tuple $(\pi, \C)$,
where $\pi = \{\pi_1, \pi_2, ..., \pi_n\}$. $\pi_i$ is the joint distribution of size and profit for item $b_i$, $\forall i$. $\C$ is the capacity of the knapsack.
W.l.o.g., we can assume $\frac{1}{2}\leq\C\leq2$ 
and the size of each item is distributed between $0$ and $2\C$.
The relaxed knapsack capacity $\C+O(\e)$ should be less than $2\C$.
The distributions for different items are mutually independent.
We let random variables $X_b$ and $P_b$ denote the size and profit of item $b$.
We use $\pi_b$ to denote the probability distribution of $X_b$, i.e.,
$\pi_b(w)=\Pr[X_b=w]$.
W.l.o.g., we can assume that for each item, we obtain a fixed profit for each realized size.
For each item $b$, we define the  \textit{effective profit}  function 
\footnote{
We find the effective profit function easier to work with than the profit function
when the size and the profit are correlated.
}
$p_b$ to be
$$
p_b(w) = \E[P_b\mid X_b=w]\cdot \Pr[X_b=w]\,\, \text{ for any } \,\,\,w\in[0, \C].
$$
We use the shorthand notation $p_{b}(I) = \sum_{w \in I}p_b(w)$ for any $I\subseteq[0, \C]$.

\vspace{0.3cm}
\topic{\bf Policies:}
The process of applying a policy $\sigma$ on an instance $(\pi, \C)$ can be represented as a decision tree $T(\sigma, \pi, \C)$.
Each node $v$ in $T_\sigma$ corresponds to placing an item in the knapsack.
Each edge $e=(v,u)$ in $T_\sigma$ ($v$ is the parent) corresponds to a size realization of $v$.
We use $w_e$, $\pi_e$ to denote the corresponding size and probability of $e$, respectively.
We also use $T_\sigma$ to denote $T(\sigma, \pi, \C)$ when the context is clear.

We call the path from root to $v$ in $T_\sigma$ the \textit{realization path} of $v$, and denote it by $R(v)$.
For a node $v$, we denote the {\em occupied capacity} before inserting $v$ as
$
W(v)=\sum_{e\in R(v)}w_e
$
and the probability of reaching $v$ as
$
\Phi(v)=\prod_{e\in R(v)}\pi_e.
$
Denote by $R^\sigma$ the random set of items that $\sigma$ packs.

We use $\P(\sigma, \pi, \C)$ to denote the expected profit that
the policy $\sigma$ can obtain with the given distributions $\pi$ and total capacity $\C$.
We also use the shorthand notation $\P(\sigma)$ or $\P(T_\sigma)$ if the context is clear.
Recursively define the expected profit of the subtree $T_v$ rooted at $v$ to be
$$
\P(v)=\sum_{ e = (v,u) \mid W(v) + w_e \leq \C }\Big[p_v(w_e)+\pi_e\cdot \P(u)\Big].
$$
The expected profit $\P(\sigma)$ of policy $\sigma$
is simply $\P(\text{the root of }T_\sigma)$.
We use $\opt$ to denote the expected profit of the optimal adaptive policy.
We note that in some steps of our algorithm, we assume the knowledge of $\opt$.
In fact, any constant approximation of $\opt$,
which for example can be obtained using the approximation algorithm in \cite{dean2008approximating} for \SK\
or the one in \cite{gupta2011approximation} for \SKCRC,
would suffice for our purpose.

As we mentioned before, the problem is PSPACE-hard and the optimal policy may be
exponentially large.
In order to reduce search space, we need to focus on a
very special class of policies, in which it is possible to find a nearly optimal
policy efficiently and this policy is also close to the optimal policy for the original problem.
We start with some simple properties shown in Bhalgat et al.~\cite{bhalgat10}
\footnote{
In fact, they only considered the basic version of stochastic knapsack where the profit
is a fixed value for an item.
However, a scrutiny of their proofs shows that correlated profits do not affect the properties.
}.
W.l.o.g., we also assume that all (optimal or near optimal) policies $\sigma$ considered in this paper
have the following property:
\begin{itemize}
\item[{ P1.}] For $u, v\in T_\sigma$, if $u$ is an ancestor of $v$, then $\P(u)\geq \P(v)$.
\end{itemize}
Otherwise, replacing the subtree $T_u$ with $T_v$ increases the profit of the policy $\sigma$.
This also implies that for any $v\in T_\sigma$, $\P(v)\leq \P(\sigma)$.

\begin{lem}[part of Lemma 2.4 in \cite{bhalgat10}]
\label{lm:basicprop}
For any policy $\sigma$ on instance $(\pi,\C)$, there exists a policy $\sigma'$ such that
$\P(\sigma', \pi, \C)=(1-O(\e))\P(\sigma, \pi, \C)$ and
\begin{enumerate}
  \item[{\em P2.}] for any realization path $R$ in $T(\sigma', \pi, \C)$, $\sum_{v\in R}\E[X_v]=O(\C/\e)$.
\end{enumerate}
\end{lem}




\subsection{Discretization}

In this section, we discuss how to discretize the size and profit distributions
for items in $\B$, using parameter $\e$.
W.l.o.g., we assume that the range of $X_b$ is $[0, 2\C)$ for any item $b$.
The discretization of the size distributions is the same as the one in Section~\ref{sec:discretization}.
We also need to discretize the profit distributions.
For each item $b$ and $w \in [0, 2\C]$, we use $D_b(w)$ to denote the discretized size of $w$ for item $b$,
i.e., $D_b(w)$ is the value of $\wt X_b$ for $X_b = w$.
The discretized effective profit function $\wt p_b$ is defined to be
$$
\wt p_b(\size) = \sum_{w \mid D_b(w) = \size} p_b(w), \text{ for all } \size\in \dsize.
$$
This finishes the description of the discretization step


%
%

We need the notion of {\em canonical policies} introduced in \cite{bhalgat10}. 
A policy $\wt\sigma$ is a canonical policy if it makes decisions based on the discretized sizes of items inserted,
rather than their actual sizes.
A canonical policy stops inserting items when the total discretized size of items inserted exceeds the knapsack capacity $\C$.
Before that, it attempts to insert items even if the total actual size overflows.
No profit from those items can be collected.
In this following lemma, we show it suffices to only consider canonical policies. 
The proof is similar to that of Lemma A.5 in \cite{bhalgat10}, which can be found in the appendix.
Due to the presence of the correlations between profits and sizes,
we need to be more careful in bounding the profit loss.

\begin{lem}
\label{lm:policytransform}
Let $\pi$ be the joint distribution of size and profit for items in $\B$ and $\wt\pi$ be the discretized version of $\pi$.
Then, the following statements hold:
\begin{enumerate}
  \item For any policy $\sigma$, there exists a canonical policy $\wt\sigma$ such that
  $$
  \P(\wt\sigma, \wt\pi, (1+4\e)\C) = (1-O(\e))\P(\sigma, \pi, \C);
  $$
  \item For any canonical policy $\wt\sigma$,
  $$
  \P(\wt\sigma, \pi, (1+4\e)\C) = (1-O(\e))\P(\wt\sigma, \wt\pi, \C).
  $$
\end{enumerate}
\end{lem}

\subsection{Block-Adaptive Policies}


To further reduce the search space,
Bhalgat et al. \cite{bhalgat10} discovered a very specific class of canonical policies, called {\em block-adaptive policies}
and showed it is sufficient to restrict the search to this set if we are satisfied 
with a nearly optimal policy.
In a block-adaptive policy,
instead of inserting one item at a time,
we insert a set of items together each time.
This set of items is called a \textit{block}.
A block-adaptive policy $\hsigma$ can also be thought as a decision tree $T_{\hsigma}$
where each node in the tree corresponds to a block.
Each edge incident on a vertex corresponds to a realization of the sum of the discretized sizes
of all items in the block.

It has been shown in \cite{bhalgat10} that
for \SK, from an optimal (or nearly optimal) adaptive canonical policy $\sigma$,
we can construct a block adaptive policy (with some other nice properties),
from which we can get almost as much profit as from $\sigma$, as in the following Lemma.

\begin{lem}
\label{lm:SKCRC_blockpolicy}
A canonical policy $\wt\sigma$ with expected profit $(1-\e)\opt$
can be transformed into a block-adaptive policy
with expected profit $(1 - O(\e))\opt$ when the capacity is further relaxed by $\e\C$.
Moreover, the block-adaptive policy satisfies the following properties:
\begin{enumerate}
  \item[B1.] There are $O(\e^{-14})$ blocks on any root-leaf path in the decision tree.
  \item[B2.] There are $|\dsize|=O(\C/\e^5)$ children for each block.
  \item[B3.] Each block $M$ with more than one items satisfies that
  $\sum_{b\in M} \E\big[\wt X_b\big] \leq \e^{13}$.
   \footnote{In fact, this property was not explicitly mentioned in Bhalget et al.~\cite{bhalgat10}.
  But it can be concluded from 
  the fact that $M$ has total profit at most $2\e^{14}\opt$
  and each item has a profit density at least $\e\opt$. In our alternative proof provided in Section~\ref{sec:SKCRC_blockadaptive},
  we do not need the notion of profit density.}
\end{enumerate}
\end{lem}

In Section~\ref{sec:SKCRC_blockadaptive}, we provide a proof for the generalization of
the above lemma to \SKCRC.

\subsection{Poisson Approximation}

To search for the (nearly) optimal block-adaptive policy, we want to enumerate all possible
structures for a block.
In \cite{bhalgat10}, this is done by enumerating all different combinations of
the profit contributions from $O(\log n)$ equivalence classes of items, using the technique developed in \cite{chekuri2000ptas}.
Instead, we enumerate all possible signatures of a block, similar to what we have done in Section~\ref{sec:fixset}.
Since we consider correlated profits and sizes, a signature needs to reflect the profit distribution as well as the size distribution.
Formally, for an item $b$, we define the {\em signature} of $b$ to be
$$
\sig(b) = \Bigl(\ol p_b(\size_0), \ol p_b(\size_1), \ldots, \ol p_b(\size_{z-1}); \ol\pi_b(\size_1), \ol\pi_b(\size_2), \ldots, \ol\pi_b(\size_{z-1})\Bigr),
$$
where $\ol p_b(\size) = \bigl\lfloor \wt p_b(\size)\cdot\frac{n}{\e^{23}\opt}\bigr\rfloor\cdot\frac{\e^{23}\opt}{n}$
and $\ol \pi_b(\size) = \bigl\lfloor \wt \pi_b(\size)\cdot\frac{n}{\e^{23}}\bigr\rfloor\cdot\frac{\e^{23}}{n}$ for any $\size\in \dsize$.
For a block $M$ of items, we define the {\em signature} of $M$ to be

\begin{center}
$\sig(M)=\sum_{b\in M}\sig(b).$
\end{center}
We denote the ``zero size probability'' of $M$ to be $\wt\pi_M^0 = \Pr \bigl[ \wt X(M) = 0 \bigr]$.

The following lemma shows that it is sufficient to enumerate all signatures for all blocks in a
block-adaptive policy.
\begin{lem}
\label{lm:signature}
Consider two decision trees $T_1, T_2$ corresponding to block-adaptive policies with the same topology
(i.e., $T_1$ and $T_2$ are isomorphic).
If for each block $M_1$ in $T_1$, the block $M_2$ at the corresponding position in $T_2$
satisfies that $\sig(M_1) = \sig(M_2)$,
then $|\P(T_1)-\P(T_2)|=O(\e)\opt$.
\end{lem}

Before proving Lemma~\ref{lm:signature}, we need to prove the following result.

\begin{lem}
\label{lm:blockreplacement}
Suppose the capacity of the knapsack is $\C'\leq C$.
Let $M_1, M_2$ be two blocks with the same signature $\sig_M$.
Let $\sigma^b$ be a block adaptive policy in which $M_1$ is the root block.
Let $\P(M_1)$ be the expected profit we can get from $M_1$ with a knapsack capacity $\C'$.
Then, replacing $M_1$ with $M_2$ in $T_{\sigma^b}$
incurs a expected profit loss of at most $O(\e^{18})\opt + O(\e^9)\P(M_1)$.
\end{lem}

\begin{proof}
For ease of notation, we use $T_1$ to denote $T_{\sigma^b}$
and $T_2$ the tree obtained by replacing $M_1$ with $M_2$.
Let $\P(M_2)$ be the expected profit we can get from $M_2$ with a knapsack capacity $\C'$.
First, we show the following two useful results.
\begin{enumerate}
  \item[\quad(a)]
    \quad$\left|\P(M_1) - \P(M_2)\right| = O(\e^{18}) = O(\e^{18})\opt + O(\e^9)\P(M_1)$.
  \item[\quad(b)]
    \quad$\Delta\left(\wt X(M_1), \wt X(M_2)\right) = O(\e^{18}).$
\end{enumerate}
It is straightforward to verify the above results
for the case where both $M_1$ and $M_2$ have only one item,
from the definition of signatures.

Now, we focus on the case where $M_1$ has more than one items. The case where $M_2$ has more than one items
is the same.
By Lemma~\ref{lm:SKCRC_blockpolicy}~B3, $\E[\wt X(M_1)] = \sum_{b \in M_1} \E[\wt X_b] \leq 2\e^{13}$.
Then we have,
$$
\E\bigl[\wt X(M_2)\bigr] = \sum_{\size >0}\Bigl(\sum_{b \in M_2} \wt\pi_b(\size)\Bigr)\cdot \size
\leq \sum_{\size >0 }\Bigl(\sum_{b\in M_1}\wt\pi_b(\size) + \e^{23}\Bigr)\cdot \size
\leq \E\bigl[\wt X(M_1)\bigr] + O(\e^{18}) \leq 3\e^{13}.
$$
By Markov's inequality, $\wt\pi_{M_1}^0= 1 - \Pr\left[ \wt X(M_1) \geq \e^4 \right] \geq 1 - 2\e^9$, and $\wt\pi_{M_2}^0\geq 1 - 3\e^9$.

Suppose we insert the items in $M_1$ one by one.
For any item $b \in M_1$, with probability at least $\wt\pi_{M1}^0$,
the remaining capacity before inserting $b$ is $\C'$ (all previous items realized to zero size).
So the expected profit we can get from $b$ is at least $\wt\pi_{M_1}^0\cdot\wt p_b\bigl([0, \C']\bigr)$.
Thus
$$
\P(M_1) \geq \wt\pi_{M_1}^0\sum_{b\in M_1}\wt p_b\bigl([0, \C']\bigr)
= \wt\pi_{M_1}^0\sum_{0 \leq \size \leq \C'}\sum_{b \in M_1}\wt p_b(\size).
$$
We also have that
$$
\P(M_1) \leq \sum_{b\in M_1}\wt p_b\bigl([0, \C']\bigr) = \sum_{0 \leq \size \leq \C'}\sum_{b \in M_1}\wt p_b(\size).
$$
Similarly, we can show that
$$
\wt\pi_{M_2}^0\sum_{0 \leq \size \leq \C'}\sum_{b \in M_2}\wt p_b(\size) \leq \P(M_2) \leq \sum_{0 \leq \size \leq \C'}\sum_{b \in M_2}\wt p_b(\size).
$$
Since $\sig(M_1) = \sig(M_2) = \sig_M$, we have that
$$
\Bigl|\sum_{0 \leq \size \leq \C'}\sum_{b \in M_1}\wt p_b(\size) - \sum_{0 \leq \size \leq \C'}\sum_{b \in M_2}\wt p_b(\size)\Bigr| = O(\e^{18})\opt.
$$
Linking these inequalities together, we obtain (a).

On the other hand, since $\E[\wt X(M_1)] \leq 2\e^{13}$,
$\E[\wt X(M_2)] \leq 3\e^{13}$
and
$\ol\pi_{M1}(s) = \ol\pi_{M2}(s)$ for any $s$,
we can show (b) holds also by following
the same proof as that of Lemma ~\ref{lm:fixsetsignature}, which we do not repeat here.

Let $v_\size$ be the child of $M_1$ corresponding to size realization $\size$,
and $T_\size$ be the subtree rooted at $v_\size$.
Given (a) and (b), we have that
\begin{eqnarray*}
\P(T_2) &=& \P(M_2) + \sum_\size\Pr\left[\wt X(M_2) = \size\right]\P(T_\size)  \\
   &\geq& \P(M_1) - O(\e^{18})\opt - O(\e^9)\P(M_1) \\
   & &+ \sum_\size\Pr\left[\wt X(M_1) = \size\right]\P(T_\size)
    - \Delta\left(\wt X(M_1) , \wt X(M_2)\right)\cdot\max_\size\P(T_\size) \\
   &=& \P(T_1) - O(\e^{18})\opt - O(\e^9)\P(M_1).
\end{eqnarray*}
The last inequality holds since P1: $\max_\size\P(T_\size)\leq \opt$.
\end{proof}

\begin{proofoflm}{\ref{lm:signature}}
We replace all the blocks in $T_{\sigma_1^b}$ by the corresponding ones in $T_{\sigma_2^b}$. By Lemma~\ref{lm:blockreplacement}, the total profit loss is at most
\begin{eqnarray*}
    & & \sum_{M} \Phi(M)\left[O(\e^{18})\opt + O(\e^9)\P_M\right] \\
    &=& O(\e^{18})\sum_{M} \Phi(M)\opt
    + O(\e^9)\sum_{M}\Phi(M)\P_M \\
    &=& O(\e^{18})\sum_{M} \Phi(M)\opt
    + O(\e^9)\P(\sigma_1^b) \\
    &\leq& O(\e)\opt
\end{eqnarray*}
The last inequality holds because $\P(\sigma_1^b)\leq\opt$ and the depth of $T_{\sigma_1^b}$ is $O(\e^{-14})$,
thus $\sum_{M} \Phi(M) = O(\e^{-14})$.
\qed
\end{proofoflm}

The number of possible signatures for a block is $\bigl(O(n/\e^{23})\bigr)^{O(\e^{-5})}=n^{\poly(1/\e)}$, which is a polynomial of $n$.
For any block adaptive policy $\sigma^b$, there are at most $f(\e) = \bigl(O(\e^{-5})\bigr)^{O(\e^{-14})}=2^{\poly(1/\e)}$ blocks in its decision tree,
since the height of the tree is $O(\e^{-14})$ and the branching factor is at most $|\dsize| = O(\e^{-5})$.
Therefore, the number of all topologies of the decision tree is a constant.



\subsection{Finding a Nearly Optimal Block-Adaptive Policy}

We have shown it suffices to
enumerate over all topologies of the decision trees
along with all possible signatures for each block
(the number of all possibilities is $n^{O(f(\e))}=n^{2^{\poly(1/\e)}}$)
in order to find a nearly optimal block-adaptive policy.
Now, we show how to find a nearly optimal block-adaptive policy with a given tree topology along with the signatures for all blocks,
using dynamic programming.

The dynamic program is fairly standard and
we present a sketch here.
Assume the tree topology has been fixed.
A configuration $\dpc$ in the dynamic program is a set of signatures,
each corresponding to a block in the tree.
As we have shown, the number of configurations is $\poly(n)$.
We use $\DP(i, \dpc)=1$ to denote the fact that
we can reach configuration $\dpc$ using a subset of $\{b_1,\ldots, b_i\}$.
Otherwise, $\DP(i, \dpc)=0$.
Initially, $\DP(0, \mathbf{0})=1$.
We compute all $\DP(i, \dpc)$ values in an lexicographically increasing order of $(i,\dpc)$.
The value of $\DP(i,\dpc)$ can be computed from the values of $\DP(i-1, \dpc')$
for all $\dpc'\leq \dpc$ (coordinatewise).
In fact, this step can be done as follows.
Suppose we want to compute $\DP(i,\dpc)$.
We can decide to place item $b_i$ in a few blocks in the decision tree.
The constraint here is no two blocks where we place $b_i$ have an ancestor-descendant relationship.
Since the size of tree is $f(\e)=(1/\e^5)^{O(1/\e^{14})}$, so the number of possible ways of
adding item $b_i$ is $2^{f(\e)}$ which is (still) a constant.
For a particular placement of $b_i$, we subtract the contribution of $b_i$
from configuration $\dpc$ (i.e., subtract $\sig(b_i)$ from the vectors in $\dpc$ corresponding to
the blocks where we place $b_i$), resulting another configuration $\dpc'$.
We let $\DP(i, \dpc)\leftarrow \max(\DP(i, \dpc), \DP(i-1, \dpc'))$.
Since the number of tree topologies is a constant, 
the number of configurations is $n^{f(\e)}$ and computing each $\DP(i,\dpc)$ values
takes a constant time, the overall running of our algorithm is $O(n^{f(\e)})$, which improves
upon the $n^{O(f(\e))^{O(f(\e))}}$ running time in \cite{bhalgat10}.

\vspace{0.2cm}

Now we have all necessary components to show the main theorem of this section.
\vspace{0.2cm}

\begin{proofofthm}{~\ref{thm:sk}}
Suppose $\sigma^*$ is the optimal policy. The optimal expected profit is denoted as
$\P(\sigma^*, \pi, \C) = \opt$.
Given an instance $(\pi, \C)$, the first step is to compute
the discretized distribution $\wt\pi$.
Then we use the dynamic program to find a nearly optimal block adaptive policy $\sigma$ for $(\wt\pi, (1+5\e)\C)$.
By result 1 of Lemma~\ref{lm:policytransform}, there exists a canonical policy $\wt\sigma$ such that
$$
\P(\wt\sigma, \wt\pi, (1+4\e)\C) \geq (1-O(\e))\P(\sigma^*, \pi, \C) = (1-O(\e))\opt.
$$
By Lemma~\ref{lm:SKCRC_blockpolicy}, there exists a block adaptive policy $\sigma^b$ such that
$
\P(\sigma^b, \wt\pi, (1+5\e)\C) \geq (1-O(\e))\opt.
$
Since the configuration of $\sigma^b$ is enumerated at some step of the algorithm,
our dynamic program is able to find a block adaptive policy $\sigma$ with the same configuration
(the same tree topology and the same signatures for corresponding nodes).
By Lemma~\ref{lm:signature}, we can see that
$$
\P(\sigma, \wt\pi, (1+5\e)\C) \geq (1-O(\e))\P(\sigma^b, \wt\pi, (1+5\e)\C) \geq (1-O(\e))\opt.
$$
By result 2 of Lemma~\ref{lm:policytransform},
$$
\P(\sigma, \pi, (1+4\e)(1+5\e)\C)
\geq (1-O(\e))\P(\sigma, \wt\pi, (1+5\e)\C) \geq (1-O(\e))\opt.
$$
Hence, the proof of the theorem is completed.
\qed
\end{proofofthm}


\section{Stochastic Knapsack with Correlations and Cancelations}
\label{sec:SKCRC}

Recall in the stochastic knapsack problem with correlations and cancelations (\SKCRC),
we can cancel a job in the middle and we gain zero profit from a canceled job.
If we decide to cancel job $b$ after running for $t$ time units,
job $b$ can be thought as a job $b^{t}$ with running time $X_{b^{t}}=\min\{X_b, t\}$, where $X_b$ is the processing time of $b$.
The effective profit of the new job $p_{b^{t}}(x)$ equals $p_b(x)$ if $x < t$ and $0$ if $x \geq t$.
Since we consider discrete time distributions, it only makes sense to cancel a job
after a discrete point with nonzero probability mass.
Suppose the size of the support of each time distribution is bounded by $m$.
Therefore,
for each job $b$, we can use a set of $m$ jobs to represent all possible cancelations of $b$,
and in each realization path, we are allowed to choose at most one job from the set.
In fact, we solve the following more general problem.
We have $n$ sets of items, $\B_i, i=1,2,...,n$. Each set $\B_i$ consists of several items $b_{i1}, b_{i2}, ...$.
Our goal is to find a policy that packs at most one item from each item set $\B_i$ to the knapsack with capacity $\C$
such that the expected profit is maximized.  We call this problem the {\em generalized stochastic knapsack} problem,
denoted as \GSK.

It is not clear how to use the technique in Bhalgat et al.~\cite{bhalgat10} to handle the above problem.
In fact, they discretize all (size) probability distributions into $q=O(\log n)$ equivalent classes
and enumerate all combinations of the profit contributions from different classes in each block.
For this purpose, they adopt the technique developed in \cite{chekuri2000ptas} to reduce the number of combinations to a polynomial.
Once the profit contribution from a class to a block is fixed, the actual set of items assigned into that block
can be easily determined in a greedy manner (items with larger profit should be packed into the blocks that are closer to the root).
However, in \SKCRC, $\B_i$ may contain items from several classes.
Assigning an item in $\B_i$ to a block would prevent us from assigning any other items in $\B_i$ to the same block and its descendants.
Such dependency makes the item assignment very complicated (for a fixed combination)
and it is not clear to us how to do this in polynomial time.
However, our technique can be easily extended to \SKCRC.

\subsection{Block-Adaptive Policies}
\label{sec:SKCRC_blockadaptive}
In this section, we show that Lemma~\ref{lm:SKCRC_blockpolicy} also holds for \GSK\ (thus also for \SKCRC).
We note that the proof in \cite{bhalgat10} does not generalize to \GSK\
and we need to modify the proof in some essential way.
We note that our proof works even when there are arbitrary precedence or cardinality constraints imposed on
the items. In fact, any realization path of the constructed block-adaptive policy corresponds to some realization path
of the original policy $\wt\sigma$. The idea of the proof may be useful in showing the existence of nearly optimal
block-adaptive policies for other problems, thus may be of independent interest.

\begin{proof}
For any node $v$ in the decision tree $T_{\wt\sigma}$, we define the
{\em leftmost path of $v$} to be the realization path which starts at $v$, ends
at a leaf, and consists of only edges corresponding to size zero.
For any node $v$ and size $\size\in \dsize$,
we use $v_{\size}$ to denote the {\em $\size$-child} of $v$, that is
the child of $v$ corresponding to the size realization $\size$.
We define the {\em segment} starting with node $v$ (denoted as $\segment(v)$) in $T_{\wt\sigma}$
as the maximal prefix of the leftmost path of $v$ such that:
\begin{enumerate}
  \item If $\E[\wt X(v)] > \e^{13}$, $\segment(v)$ is the singleton node $\{v\}$. Otherwise, $\E[\wt X(\segment(v))] \leq \e^{13}$;
  \item For any two nodes $ u, w \in \segment(v)$, and any size $\size$,
  $\left|\P(T_{u_\size}) - \P(T_{w_\size})\right| \leq \e^5\opt$.
\end{enumerate}

We partition $T_{\wt\sigma}$ into segments as follows.
We say a node $v$ is a {\em starting node} if $v$ is a root or $v$ corresponds to a non-zero size realization of its parent.
For each starting node $v$, we greedily partition the leftmost path of $v$ into segments, i.e.,
delete $\segment(v)$ and recurse on the remaining part.
Fix a particular root-to-leaf path $R$.
Let us bound the number of segments on $R$.
Suppose we are at node $u$, which is a node in $\segment(v)$, and the next node we are about to visit is $w$
which is not in $\segment(v)$.
We know that one of the following events must happen:
\begin{enumerate}
\item[1.] $w$ corresponds to a non-zero size realization of $u$;
\item[2.] $\sum_{u'\in \segment(v)\cup\{w\}}\E\bigl[\wt X(u')\bigr] > \e^{13}$;
\item[3.] For some size $\size$, $|\P(T_{v_\size}) - \P(T_{w_\size})| > \e^5\opt$.
\end{enumerate}
The first event happens for at most $O(\C/\e^{4})$ times.
Since $\sum_{v\in R}\E[\wt X(v)] = O(\C/\e)$ by Lemma~\ref{lm:basicprop},
the second event happens at most $O(\e^{-14})$ times.
Suppose $R=\{v^1,v^2,\ldots, v^k\}$.
W.l.o.g., we can assume that for each size $s$, $\opt\geq \P(T_{v^1_\size}) \geq \P(T_{v^2_\size})\geq \ldots\geq \P(T_{v^k_\size})\geq 0$
by a simple substitution argument similar to P1.
For each particular size $\size$, the third event occurs for $O(\e^{-5})$ times.
Since there are at most $|\dsize|=O(\e^{-5})$ different sizes,
we need at most $|\dsize|\cdot\e^{-5} = O(\e^{-10})$ parts.
This gives a bound $O(\e^{-14}+\e^{-10}+\e^{-4})=O(\e^{-14})$ on the number of segments on each root-leaf path.

Now, we are ready to describe the algorithm,
which takes a canonical policy $\wt\sigma$ as input, and packs items into the knapsack in a block-adaptive way.
For any segment $\segment(v)$, we use $l(v)$ to denote the last node in $\segment(v)$.
Similar to the argument in \cite{bhalgat10},
we use two knapsacks, the main knapsack with capacity $\C$ and the auxiliary knapsack with capacity $\e\C$.

\IncMargin{1em}
\begin{algorithm}[H]

\DontPrintSemicolon

\nlset{1.} Initially, $S = \emptyset$. $S$ represents the auxiliary knapsack. \;
\nlset{2.} We start at the root of $T_{\wt\sigma}$. \;
\Repeat {
    \em A leaf in $T_{\wt\sigma}$ is reached.
}
{
    \nlset{3.1} Suppose we are at node $v$ in $T_{\wt\sigma}$. Add the items in $\segment(v)$ to the main knapsack one by one
    until some node $u$ realizes to a nonzero size, say $\size$. \;
    \nlset{3.2} Add all remaining items in $\segment(v)$ to the auxiliary knapsack $S$. \;
    \nlset{3.3} Visit node $l(v)_\size$, the $\size$-child of the last node of $\segment(v)$. \;
    \nlset{3.4} If all nodes in $\segment(v)$ realize to size $0$, visit $l(v)_0$. \;
}

\nlset{4.} If the auxiliary knapsack overflows, discard the entire profit. \;
\label{algo:SKCC_blockpolicy}

\PrintSemicolon

\end{algorithm}
\DecMargin{1em}

We can see that the set of items the algorithm attempts to insert always
corresponds to some realization path in the original policy $\wt\sigma$.
Hence, the algorithm packs at most  one item from each item set $\B_k$.
We have shown that B1 holds.
Properties B2 and B3 are straightforward from the definition of segments and the algorithm.

Now we show that the expected profit that new policy can obtain is at least $(1-O(\e))\opt$.
Let us focus on the profit we collect from the main knapsack and ignore those from the auxiliary knapsack.
We note that the main knapsack never overflows.
Our algorithm deviates the policy $\wt\sigma$ whenever some node $u$ in the middle of some segment $\segment(v)$
realizes to a nonzero size, say $\size$.
In this case, $\wt\sigma$ would visit $u_\size$, the $\size$-child of $u$
and follows $T_{u_\size}$ from then on,
but our algorithm visits $l(v)_\size$, the $\size$-child of the last node of that segment,
and follows $T_{l(v)_\size}$.
The expected profit loss in each such event can be bounded by $|\P(T_{u_\size})-\P(T_{l(v)_\size})|\leq \e^{5}\opt$.
Suppose $\wt\sigma$ pays such a profit loss, and switches to visit $l(v)_\size$.
Hence, $\wt\sigma$ and our algorithm always stay at the same node.
Note that the number of edges corresponding to nonzero size realizations is at most $\C/\e^4$ in any root-to-leaf path.
So $\wt\sigma$ pays at most $O(\e^{-4})$ times in any realization.
Therefore, the total profit loss is at most $O(\e \opt)$.

Finally, we note that we may not be able to collect the profit from the main knapsack in every realization since
the auxiliary knapsack may overflow, in which case we lose all the profit.
An easy (but important) observation is that the decision of the policy
is independent of the size realizations of the items in the auxiliary knapsack $S$.
Therefore, we can think that the sizes of the items in $S$ are realized after the execution of the algorithm.
If $S$ overflows, we lose the entire profit.
Since there are at most $\C/\e^4$ items realizes to a nonzero size,
$S$ is packed with items from at most $\C/\e^4$ segments.
By the first property of a segment, $\E[\wt X(S)] \leq \e^9\C$.
By Markov's inequality, $\Pr[\wt X(S) > \e\C ] \leq \e^8$.
Hence, the expected profit we gain is at least $(1-\e^8)(1-O(\e))\opt$.
\eat{
Note that the size and profit of an item only realizes after it is packed into a knapsack,
we have that
\begin{eqnarray*}
  \P(\sigma^b, \wt\psi, (1+\e)\C) &=& \sum_S \Pr[S_a = S] \cdot \Pr\left[\wt X(S) \leq \e\C\right]
   \cdot \left[\P(\sigma', \wt\psi, \C) \mid S_a = S \right] \\
   &\geq& \left(1-2\e^8\right) \sum_S \Pr[S_a = S]
   \cdot \left[\P(\sigma', \wt\psi, \C) \mid S_a = S \right] \\
   &=& (1-2\e^8)\P(\sigma', \wt\psi, \C),
\end{eqnarray*}
where $\P(\sigma', \wt\psi, \C) \mid S_a = S$ is the expected profit that $\sigma'$ will get condition on $S_a = S$.

Therefore, $\P(\sigma^b, \wt\psi, (1+\e)\C) = (1-O(\e))\P(\wt\sigma, \wt\psi, \C)$.
}
\end{proof}

\subsection{Finding a Nearly Optimal Block-Adaptive Policy}
\label{sec:SKCRC_DP}
We need to modify the dynamic program to incorporate the constraint that at most one item from each $\B_i$ can be packed.
Now, we use $\DP(i, \dpc)=1$ to denote the fact that
we can reach configuration $\dpc$ using items from $\{\B_1,\ldots, \B_i\}$, such that
on each realization path, at most one item in $\B_i$ can appear at most once.
Suppose we want to compute the value of $\DP(i, \dpc)$ for some $i$ and configuration $\dpc$.
Suppose we have computed the values of $\DP(i-1,\dpc')$ for all $\dpc'\leq \dpc$.
We can decide to place items from $\B_i$ in a few blocks in the decision tree.
The only constraint here is that
we can place at most one item from $\B_i$ in each realization path.
So the number of ways to do so is bounded by
$m^{f(\e)}$ for $f(\e)=2^{(1/\e^5)^{O(1/\e^{14})}}$,
which is a polynomial of the input size.
So the overall running time of the dynamic program is still a polynomial.
The rest of the analysis is the same as before and we do not repeat it here.

In summary, we have the following theorem, from which Theorem~\ref{thm:skcc}
follows as a direct corollary.

\begin{thm}
\label{thm:gsk}
For any $\e>0$, there is a polynomial time algorithm that finds a $(1+\e)$-approximate
adaptive policy for \GSK\ when the capacity is relaxed to $1+\e$.
\end{thm}

\vspace{0.2cm}
\topic{\SKCRC\ with Exponential Number of Realizations}
We have assumed that the size distribution of each item has a polynomial size support.
Now, we consider the case where the supports may be of exponential sizes (must be represented implicitly).
The only assumption we make here is that for any item $b_i$ and time $t$, we can compute 
the signature of $b_i^t$ in polynomial time. 
The catch is that even though there are exponential even infinite realizations, 
the number of possible signatures is bounded by a polynomial.
Moreover, as we increase $t$, each coordinate of the signature of $b_i^t$ changes monotonically,
thus the same signature does not appear again.
Hence, starting from $t=0$, we can use binary search to identify the first point
where the signature changes. Repeating this, we can find all different signatures for $b$, 
each corresponding to $b^t$ for some $t$,
in polynomial time.
%
In the dynamic program, we only care the signature of an item, instead of an item per se.
So we can let $\B_i$ contain only those $b_i^t$s that correspond to distinct signatures.
The size of $\B_i$ is therefore bounded by a polynomial.
                
\vspace{0.2cm}
\topic{\SKCRC\ without Relaxing the Capacity}
Combining this result and the algorithm developed in~\cite{bhalgat20112}
we can give a $(2+\e)$-approximation for \SKC\ (Theorem~\ref{thm:norelax1}).
This improves the factor 8 approximation algorithm developed in~\cite{gupta2011approximation}.
We note the algorithm in~\cite{bhalgat20112} does not work for correlated sizes and profits.
So whether there is $(2+\e)$-approximation for \SKCRC\ is still open.
However, with mild assumptions on the size distributions,
we can achieve an approximation factor of $2$ for \SKCRC\ (Theorem~\ref{thm:skccnorelax}).
The details can be found in Appendix~\ref{sec:norelax}.

\section{Bayesian Online Selection}
\label{sec:ssp}
In this section, we consider the {\em Bayesian online selection problem} subject to a knapsack constraint
(denoted as \SSP). Our problem falls into the framework formulated in \cite{kleinberg2012matroid}.
In \SSP, we are given a set of items $\B = \{b_1, b_2, \ldots, b_n\}$
and a knapsack capacity $\C$. Each item $b_i$ has a random size $X_i$ and a random profit $P_i$.
$X_i$ and $P_i$ can be correlated but
different items are independent of each other.
The (discrete) joint distribution $\pi$ of $X_i$ and $P_i$ is the input to the problem,
in the form of $\pi_i(x,p) = \Pr(X_i=x, P_i=p)$.
Let $\calD_i$ be the support of $\pi_i$, i.e., $\calD_i=\{(x,p) \mid \pi_i(x,p)\ne 0\}$.
An adaptive policy $\sigma$ can choose an item $b_i \in \B$ each time, view the size realization of $X_i$,
and then make a irrevocable decision whether to pack $b_i$ or not.
If $\sigma$ decides to pack $b_i$ into the knapsack (given the remaining capacity is sufficient),
the profit $P_i$ is collected.
Otherwise, no profit is collected, the remaining capacity does not change and we can not
recall $b_i$ later.
If the items arrive in a predetermined order, we call the problem the {\em
fixed order \SSP} problem. 

Our problem is closely connected to the {\em knapsack secretary problem}~\cite{babaioff2007knapsack}
in the following sense.
In the knapsack secretary problem,
the size and the profit of each item are unknown in advance and the items arrive in a random order.
When an item arrives, its size and profit become known to the decision maker and an irrevocable decision has to be made.
Therefore, if we intentionally forget about the stochastic information and process the items in
a random order, \SSP\ becomes exactly the knapsack secretary problem, for which a constant factor
competitive algorithm is known~\cite{babaioff2007knapsack}
\footnote{
We note that most work on
secretary problems and prophet inequalities
measures the performance of the algorithm by comparing the solution found by the online algorithm
against the offline optimum, while the approximation ratios in this paper
are computed by comparing against the best adaptive policy.
}.
The ability that we can adaptively choose the order of the items in \SSP\
is very similar to the {\em free order models} of 
matroid prophet inequality~\cite{chawla2010multi} and
matroid secretary~\cite{jaillet2012advances}.

In the remainder of this section, we focus on proving Theorem~\ref{thm:ssp}.
First, we consider as a warmup an interesting special case where the profit of each item is a fixed value.
In this case, we can show the following intuitive lemma. The proof can be found in the appendix.
\begin{lem}
\label{lm:ssp}
Let $\sigma$ be the optimal policy.
Suppose $\sigma$ chooses $b$ as the next item to consider.
and decides to discard $b$ if $X_b$ realizes to $t$.
Then $\sigma$ should discard $b$ if $X_b$ realizes to a larger size $s \geq t$.
\end{lem}

\eat{
\begin{proof}
We use $\P(\sigma, \B, \C')$ to denote the expected profit the policy $\sigma$
can achieve with item set $\B$ and remaining capacity $\C'$.
Since $\sigma$ discards $b$ when $X_b$ realizes to $t$, we have that
$\P(\sigma, \B\backslash\{b\}, \C') \geq P_b + \P(\sigma, \B\backslash\{b\}, \C'-t).$
We also have $\P(\sigma, \B\backslash\{b\}, \C'-t) \geq \P(\sigma, \B\backslash\{b\}, \C'-s)$ since $\C'-t \geq \C'-s$.
Thus, $\P(\sigma, \B\backslash\{b\}, \C') \geq P_b + \P(\sigma, \B\backslash\{b\}, \C'-s)$.
Therefore, $\sigma$ should discard $b$ when $X_b = s \geq t$,
in order to maximize the expected profit.
\end{proof}
}

The above lemma suggests that at a particular stage, for an item $b_i$, there is a cutoff point $t$
such that  we accept $b$ if $X_i\leq t$, and reject $b_i$ otherwise.
In this case, item $b_i$ is equivalent to an item
with the size $X_i^t$ and profit $P_i^t$, jointly distributed as follows:
$$(X_i^t, P_i^t)=
\left\{
  \begin{array}{ll}
    (X_i, P_i), & \hbox{$X_i \geq t$;} \\
    (0,0), & \hbox{$X_i > t$.}
  \end{array}
\right.
$$
In fact, this viewpoint allows us to reduce \SSP\ to \GSK.
For each item $b_i$, we create a set of items $\B_i=\{b^t_i\}_t$
where $b^t_i$ represents the item $b_i$ with cutoff point $t$.
The only requirement is that at most one item from $\B_i$ can be packed in the knapsack.
Since we assume discrete distributions, there are at most a polynomial number cutoff points.
Hence, the size of \GSK\ instance we create is also bounded by a polynomial.
Theorem~\ref{thm:ssp} directly follows from Theorem~\ref{thm:gsk}.

Now, we consider the general case where $X_i$ and $P_i$ are correlated.
In this case, there is no single cutoff point as before.
However, we can still reduce the problem to \GSK.
Suppose we decide to consider $b_i$ at a particular stage,
and decide to accept $b_i$ if and only if the realization $(X_i, P_i) \in D$
for some $D\subseteq \calD_i$. We call $D$ the {\em acceptance set}.
Then, $b_i$ is equivalent to an item
with the size $X_i^D$ and profit $P_i^D$, jointly distributed as follows:
$$(X_i^D, P_i^D)=
\left\{
  \begin{array}{ll}
    (X_i, P_i), & \hbox{$(X_i, P_i) \in  D$;} \\
    (0,0), & \hbox{Otherwise.}
  \end{array}
\right.
$$
However, the \GSK\ instance created can be exponential in size
since each subset $D\subseteq \calD_i$ corresponds to a distinct item in $\B_i$.
We may reduce the size of $\B_i$ by exploiting the simple observation that in any optimal policy, if $(x, p)\in D$, any $(x', p')$
with $x'\leq x$ and $p'\geq p$ (a realization with a smaller size and a larger profit)
must be in $D$ (we can use the same proof as Lemma~\ref{lm:ssp} to show this).
But the resulting size is still exponential.

Now, we sketch an algorithm that reduces the size of $\B_i$ to a polynomial and incurs a profit loss of at most $O(\e \opt)$.
We modify the distribution $\pi_i$ as follows.
\begin{enumerate}
\item
For any $(x,p)\in \calD_i$ such that $x\leq \frac{\e\C}{n}$,
we move the probability mass at point $(x,p)$ to point $(0,p)$.
This step affects the total size of all packed items by at most $\e\C$ since there are at most $n$ items.
\item
For any $(x,p)\in \calD_i$ such that $p\leq \frac{\e\opt}{n}$,
we move the probability mass at point $(x,p)$ to point $(x,0)$.
This step affects the total profit by at most $\e\opt$.
\item
For other $(x,p)\in \calD_i$,
we move the probability mass at point $(x,p)$ to point $(x',p')$
where $x'$ is the largest value of the form $\frac{\e\C}{n} (1+\e)^k$ (for some $k\in \mathbb{Z}^+$) that is at most $x$,
and $p'$ is the largest value of the form $\frac{\e\opt}{n} (1+\e)^l$ (for some $l\in \mathbb{Z}^+$) that is at most $p$.
Obviously, this step affects the total size and total profit by at most a multiplicative factor $1+\e$.
\end{enumerate}

\begin{figure*}[t]
    \begin{center}
    \includegraphics[width=0.3\linewidth]{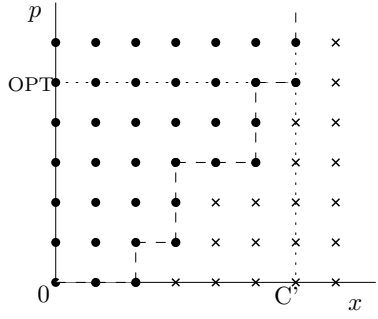}
    \caption{The staircase structure of an acceptance set $D$ in an optimal policy. The solid points are in $D$ and the crossed points are not.}
    \vspace{-0.7cm}
    \label{fig:staircase}
    \end{center}
\end{figure*}

Now, the number of different sizes is at most $\log_{1+\e}n$
and the number of different profits is at most $\log_{1+\e} \frac{np_{\max}}{\opt}$
where $p_{\max}$ is the maximum possible profit of any realization.
We first assume $p_{\max}$ is at most $\opt$.
In this case, the size of $\calD_i$ is at most $\log_{1+\e}^2n=O(\log^2 n)$
and the number of all subsets is $2^{O(\log^2n)}=n^{O(\log n)}$, which is slightly larger than a polynomial.
To reduce this number to a polynomial, the simple observation we made before comes in handy.
That is if an optimal policy accept a realization, it must accept any realization with a smaller size and a larger profit.
Think all points in $\B_i$ are arranged as $\log_{1+\e}n\times \log_{1+\e}n$ grid points, in the 2d plane.
If we include a point $(x,p)$ in $D$ (i.e., we accept the point), we need to include all points to the upper left of $(x,p)$.
Hence, the points in $D$ form a staircase structure.
The number of such structures are bounded by the number of monotone paths from the origin to the upper right corner
of the grids, which is
$$
{\log_{1+\e}n+\log_{1+\e}n \choose \log_{1+\e}n} \leq \frac{(2\log_{1+\e}n)^{\log_{1+\e}n}}{(\log_{1+\e}n/e)^{\log_{1+\e}n}}
\leq (2e)^{\log_{1+\e}n}= \poly(n).
$$
We have reduced the size of $\B_i$, thus the size of the \GSK\ instance, to a polynomial.

Now, we consider the case where $p_{\max}>\opt$,
Consider a particular item $b_i$. Suppose for some realization $(x,p)\in \calD_i$, $p> \opt$.
We call $(x,p)$ a {\em huge profit realization}.
A simple observation is that if a huge profit realization $(x,p)$ is realized and $x$ is no more than the remaining capacity $\C'$,
we should accept $(x,p)$.
Moreover, we can not accept any $(x,p)$ such that $x>\C'$.
Hence, our search range is restricted in $[0, \C']\times [0,\opt]$. See Figure~\ref{fig:staircase}.
For each $\C'$, this reduces to the previous case.
Since we only need to consider $\log_{1+\e}n$ different $\C'$ values,
the number of different acceptance sets is clearly bounded by a polynomial.
Therefore, by Theorem~\ref{thm:gsk}, we complete the proof of Theorem~\ref{thm:ssp}.

\vspace{0.2cm}
\topic{Fixed order \SSP:}
For the fixed order model, the above reduction also works.
The only change is that in the \GSK\ instance we reduce to,
we are required to examine the items in a fixed order.
Lemma~\ref{lm:policytransform} and Lemma~\ref{lm:SKCRC_blockpolicy} also hold in this case.
So, it suffices to find a nearly optimal block-adaptive policy.
We can also modify the dynamic program in Section~\ref{sec:SKCRC_DP}
to find a block-adaptive policy subject to a particular order, as follows.
Let $\DP(i, \dpc)=1$ to denote whether
we can reach configuration $\dpc$ using items from $\{\B_1,\ldots, \B_i\}$, exact one item from each $\B_i$,
such that the block where the item from $\B_j$ is placed should be no lower (w.r.t. the decision tree)
than the block where the item from $\B_{k}$ is place for any $j<k$.
When we want to compute $\DP(i,\dpc)$,
we can only place items from $\B_i$ in the lowest blocks (w.r.t. the decision tree) with non-zero signatures
(these blocks can be directly determined from $\dpc$).
The number of such placements is clearly bounded by a polynomial since
the number of blocks is a constant.
The rest is the same as in Section~\ref{sec:SKCRC_DP}.

Finally, we note that we can also get a constant competitive algorithm when compared with the offline 
optimum, using simple LP techniques \cite{li2012knapsack}
\footnote{
We can call such result a {\em knapsack prophet inequality}.
}. The algorithm can provide the information of $\opt$, up to a constant factor,
that is needed in the discretization.

\section{Concluding Remarks}

We develop the Poisson approximation technique and successfully apply it in finding approximate
solutions for a variety of stochastic combinatorial optimization problems.
These problems range from fixed set optimization problems
to adaptive online optimization problems,
from problems with a single probabilistic objective function,
to problems with many probabilistic constraints.
Our technique is conceptually simple, easy to apply, and has led to simplifications, generalizations
and/or improvements of several previous results.

Our technique also seems quite flexible and could be potentially combined with other techniques to yield 
new results for other stochastic optimization problems.
For example, we could first apply the technique in \cite{bhalgat10} to discretize the distributions
into $O(\log n)$ equivalent classes. This could further reduce the number of possible signatures,
which might be essential for problems exhibiting more complex combinatorial structures.

In the realm of approximating the distributions of the sums of $n$ random variables,
the Poisson approximation theorem and its relatives \cite{bookbarbour} work most effectively in the regime 
where, roughly speaking,  the sum of the expected values of those random variables stays constant as $n$ increases.
This is quite different from the Gaussian approximations
(e.g., CLT, Chernoff bounds, Berry-Esseen type inequalities) which typically require 
the sum of the expected values increases as $n$.
In fact, the theory of Poisson approximation is an important area in probability theory~\cite{bookbarbour}, but
it has not been explored and utilized in algorithmic applications as extensively.
We believe our technique, and more generally the theory of Poisson approximation, can find wider applications
in stochastic combinatorial optimization and other domains.

\section*{Acknowledgements}
We would like to thank Yinyu Ye and Uri Zwick
for stimulating discussions.

\bibliographystyle{plain}
\bibliography{stochastic}

\appendix

\section{Discretizing small size region}
More formally, let $0\leq d\leq \e^4$ be the value such that
$
\Pr[X_b \leq d \mid X_b \leq \e^4] \cdot \e^4 \geq \E[X_b \mid X_b \leq \e^4]
$
and
$
\Pr[X_b \geq d \mid X_b \leq \e^4] \cdot \e^4 \geq \E[X_b \mid X_b \leq \e^4]
$
(note that such $d$ must exist).
For ease of presentation, we can imagine the size value $d$ being two distinguishable values,
$d_1=d-\frac{1}{\infty}$ and $d_2=d+\frac{1}{\infty}$,
and we let $d_1<d<d_2$.
The probability mass $\Pr[X_b=d]$ also consists of two parts
$\Pr[X_b=d_1]$ and $\Pr[X_b=d_2]$
and the value are so set that
$
\Pr[X_b \leq d_1 \mid X_b \leq \e^4] \cdot \e^4 = \E[X_b \mid X_b \leq \e^4]
$
and
$
\Pr[X_b \geq d_2 \mid X_b \leq \e^4] \cdot \e^4 = \E[X_b \mid X_b \leq \e^4]
$.
Therefore, in this new item, we can use value $d$ as the threshold value.
Note that
the effective profit $p_b(d)$ should also be divided accordingly:
$p_b(d_1)=p_b(d)\Pr[X_b=d_1]/ \Pr[X_b=d]$
and
$p_b(d_2)=p_b(d)\Pr[X_b=d_2]/ \Pr[X_b=d]$.

\section{Missing Proofs}

{\bf Lemma~\ref{lm:boundexp}}
Suppose each item $b \in \B$ has a non-negative random weight $X_b$
taking values from $[0,\alpha\C]$ for some constant $\alpha\geq 1$.
Then,
for any $S \subseteq \B$ and  any $\frac{1}{2} > \e > 0$,
if $\E[\mu(X(S))] \geq \e$, then
$\E\left[X(S)\right] \leq 3\alpha/\e$.

\begin{proof}
Since $X_b$s are independent, $\Var[X(S)] = \sum_{b\in S}\Var[X_b]$.
As $X_b \in [0,\alpha\C]$, we have
$$\Var[X_b]\leq\E[X_b^2]\leq \alpha\C\cdot\E[X_b].$$
So, $\Var[X_S] \leq \sum_{b\in S}\alpha\C\cdot\E[X_b] = \alpha\C\cdot\E[X(S)]$.
Suppose for contradiction that $\E[\mu(X(S))]  \geq \e$ and $\E[X(S)] > 3\alpha/\e$.
Then, we can see that
\begin{eqnarray*}
  \Var[X(S)] &>& \Pr[X(S) < \C] \cdot \left(\E[X(S)]-\C\right)^2 \\
  &\geq& \E[\mu(X(S))] \cdot \left(\E[X(S)]-\C\right)^2 \\
   &\geq& \e\cdot\left(\E[X(S)]-\C\right)^2 \\
   &=& \e\left(\E[X(S)] - 2\C + \C^2/\E[X(S)]\right)\E[X(S)] \\
   &>& \e(3\alpha/\e - 2\C + \e\C^2/3)\E[X(S)] \\
   &>& (3\alpha-2\e\C)\E[X(S)]
   > \alpha\C\cdot\E[X(S)],
\end{eqnarray*}
which contradicts the fact that $\Var[X(S)]\leq\alpha\C\E[X(S)]$.
\end{proof}

\noindent
{\bf Lemma~\ref{lm:fixdiscretize}}
Let $S$ be a set of items such that $\E\left[X(S)\right] \leq 3/\e$.
For any $0 \leq \beta \leq \C$, we have that
\begin{enumerate}
\item $\Pr[X(S) \leq \beta] \leq \Pr[\wt X(S) \leq \beta + \e] + O(\e)$;
\item $\Pr[\wt X(S) \leq \beta] \leq \Pr[X(S) \leq \beta + \e] + O(\e)$.
\end{enumerate}

\begin{proof}
We prove the lemma for each step of discretization.
To avoid using a lot of notations,
when the context is clear,
we always use $X_b$ to denote the size of $b$ before a particular discretization step
and $\wt X_b$ to denote the size after that step.

\vspace{0.3cm}
\noindent
{\bf Step 1:}
Let $\delta_b = \wt X_b - X_b$. By our discretization,
$\E[\delta_b] =\E[\wt X_b]-\E[X_b]= 0$. Moreover,
\begin{eqnarray*}
  \Var[\delta_b] &=& \E[\delta_b^2] - \E^2[\delta_b] = \E[\delta_b^2]\\
   &=& \Pr[X_b \leq \e^4] \cdot \E[(\wt X_b - X_b)^2 \mid X_b \leq \e^4] \\
   &\leq& \E[(\wt X_b)^2 \mid X_b \leq \e^4] + \E[(X_b)^2 \mid X_b \leq \e^4]\\
   &\leq& \e^4(\E[\wt X_b] + \E[X_b])
   \leq 2\e^4\E[X_b]
\end{eqnarray*}

Let $\delta(S) = \sum_{b \in S}\delta_b$. By linearity of expectation, $\E[\delta(S)] = 0$.
As $X_b$s are independent,
$
\Var[\delta(S)] = \sum_{b \in S} \Var[\delta_b] \leq 2\e^4\E[X(S)] \leq 6\e^3.
$
Therefore, the first inequality can be seen as follows:
\begin{eqnarray*}
  \Pr[X(S) \leq \beta] &=& \Pr[X(S) \leq \beta \wedge \delta(S) \leq \e ]
   + \Pr[X(S) \leq \beta \wedge \delta(S) > \e] \\
   &\leq& \Pr[\wt X(S) \leq \beta + \e] + \Pr[\delta(S) > \e] \\
   &\leq& \Pr[\wt X(S) \leq \beta + \e] + \Var[\delta(S)] / \e^2 \\
   &\leq& \Pr[\wt X(S) \leq \beta + \e] + 6\e
\end{eqnarray*}
The proof for the second inequality is essentially the same and omitted here.%

\vspace{0.3cm}
\noindent
{\bf Step 2:}
Noting that for step 2 we have $\wt X_b \leq X_b \leq (1+\e) \wt X_b$, the lemma is obviously true.

\eat{
\vspace{0.3cm}
\noindent
{\bf Step 3:}
The total variation distance between $X_b$ and $\wt X_b$ is
$\Delta(\wt X_b, X_b) < 4\e/n^2$.
As the number of items in $S$ is at most $n$,  we have that
$$
\Delta(\wt X(S), X(S))\leq \sum_{b\in S} \Delta(\wt X_b, X_b) < 4\e/n < \e.
$$
Therefore, we have
$\Pr[X(S) \leq \alpha] \leq \Pr[\wt X(S) \leq \alpha] + \e$ and
$\Pr[\wt X(S) \leq \alpha] \leq \Pr[X(S) \leq \alpha] + \e$
since $\Delta(X(S), \wt X(S))$ is an upper bound of
$|\Pr[X(S) \in S] -\Pr[\wt X(S) \in S]|$ for any $S\subseteq \mathbb{R}$.
}

This completes the proof of the lemma.
\end{proof}

\noindent
{\bf Lemma~\ref{lm:policytransform}} \,\,\,
Let $\pi$ be the joint distribution of size and profit for items in $\B$ and $\wt\pi$ be the discretized version of $\pi$.
Then, the following statements hold:
\begin{enumerate}
  \item For any policy $\sigma$, there exists a canonical policy $\wt\sigma$ such that
  $$
  \P(\wt\sigma, \wt\pi, (1+4\e)\C) = (1-O(\e))\P(\sigma, \pi, \C);
  $$
  \item For any canonical policy $\wt\sigma$,
  $$
  \P(\wt\sigma, \pi, (1+4\e)\C) = (1-O(\e))\P(\wt\sigma, \wt\pi, \C).
  $$
\end{enumerate}

\begin{proof}
For the first result, we first prove that there is a randomized canonical policy $\sigma_r$ such that
$
\P(\sigma_r, \wt\pi, (1+4\e)\C) = (1-O(\e))\P(\sigma, \pi, \C).
$
Thus such a deterministic policy $\wt\sigma$ exists.

In the decision tree $T(\sigma, \pi, \C)$,
each edge $e = (v, u)$ corresponds to an actual size realization of item $v$.
We use $\wt w_e$ to denote the discretized size of $w_e$, i.e., $\wt w_e = D_v(w_e)$.

The randomized policy $\sigma_r$ is derived from $\sigma$ as follows.
$T_{\sigma_r}$ has the same tree structure as $T_\sigma$.
If $\sigma_r$ inserts an item $b$ and observes a discretized size $\size\in \dsize$,
it chooses a random branch in $\sigma$ among those sizes that are mapped to $\size$, i.e., $\{w\mid D_b(w)=\size\}$
according to the probability distribution $\Pr[\text{branch } w \text{ is chosen}]=\pi_b(w)/\wt\pi_b(\size)$,
where $\wt\pi_b$ is the discretized version of $\pi_b$.
We can see that the probability of an edge in $T_{\sigma_r}$ is the same as
that of the corresponding edge in $T_{\sigma}$. The only difference is two edges are labeled with different
lengths ($w_e$ in $T_{\sigma}$ and $\wt w_e$ in $T_{\sigma_r}$).

%

Now, we bound the profit we can collect from $T_{\sigma_r}$ with a knapsack capacity of $(1+4\e)\C$.
Recall $R(v)$ is the realization path from the root to $v$ in $T_{\sigma}$.
Let $\wt W(v) = \sum_{e \in R(v)}\wt{w}_e$.
W.l.o.g., we assume that
$\forall v\in T(\sigma, \pi, \C)$, $W(v) \leq \C$.
By definition, we have that
$\wt p_v(\size) = \sum_{e=(v,u) \mid \wt w_e = \size}p_v(w_e)$.
By our construction, we have that
\begin{eqnarray*}
  \P(\sigma_r, \wt\pi, (1+4\e)\C) &=& \sum_{v, s \mid \wt W(v) + s \leq (1+4\e)\C} \Phi(v) \wt p_v(s) \\
   &=& \sum_{e = (v, u) \mid \wt W(v) + \wt w_e \leq (1+4\e)\C} \Phi(v) p_v(w_e) \\
   &=& \sum_{e = (v, u) \mid \wt W(u) \leq (1+4\e)\C} \Phi(v) p_v(w_e) \\
   &=& \P(\sigma, \pi, \C)
    - \sum_{e = (v, u) \mid \wt W(u) > (1+4\e)\C} \Phi(v) p_v(w_e) \\
   &\geq& \P(\sigma, \pi, \C) \biggl[1 - \sum_{v \in S_1}\Phi(v)\biggr]
\end{eqnarray*}
where
$S_1 = \{v\in T(\sigma, \pi, \C) \mid \wt W(v) \leq (1+4\e)\C \text{ and } \exists e = (v, u), \wt W(u) > (1+4\e)\C\}$.
The last inequality holds due to Property P1.
We upper bound
$\sum_{v \in S_1}\Phi(v)$ in the following lemma.

\vspace{0.2cm}
\noindent{\bf Lemma}
Let $\sigma$ be a policy and let $(\pi, \C)$ be an instance of the stochastic knapsack problem.
Let $S$ be a set of nodes in $T(\sigma, \pi, \C)$
that contains at most one node from each root-leaf path. Then we have that
$$
\sum_{v \in S \mid \left|W(v) - \wt W(v)\right| \geq \e(\C+1)}\Phi(v) = O(\e).
$$

\begin{proof}
We can assume w.l.o.g. that $S$ is the set of leaves in $T_\sigma$.
We prove this lemma for both steps of the discretization in Section~\ref{sec:discretization}.
First, we show after
{\bf Step 1:}
it holds that
\begin{description}
  \item[(a)]
    $\sum_{v \in S \mid \left|W(v) - \wt W(v)\right| \geq \e}\Phi(v) = O(\e).$
\end{description}
For $u\in T_\sigma$, let $b_u$ be the item corresponding to $u$ and define $\delta_u = X_{b_u} - \wt X_{b_u}$.
As we have shown in Lemma~\ref{lm:fixdiscretize},
$\E[\delta_u] = 0$ and $\Var[\delta_u] \leq 2\e^4\E[X_{b_u}]$.

Let $R_\sigma$ be the random path (consisting of both nodes and edges) $\sigma$ would choose.
For $u \in T_\sigma$,
let $R_{u+}$ be the random path $\sigma$ would choose after reaching $u$ (including $u$ and edges incident on $u$).
Define $\delta(R_{u+}) = \sum_{e \in R_{u+}} (w_e - \wt w_e)$.
Since $\sigma$ packs an item $b$ before it realizes to a particular size, we have that:
$
\E[\delta(R_{u+})] = \sum_{u' \in T_u}\Pr[\sigma \text{ chooses } u'\mid \sigma\text{ chooses} u]\E[\delta_{u'}] = 0.
$
Moreover,
\begin{eqnarray*}
  \Var\left[\delta(R_{u+})\right] &=& \E\left[(\delta(R_{u+}))^2\right] - (\E[\delta(R_{u+})])^2
    = \E\left[(\delta(R_{u+}))^2\right] \\
   &=& \sum_{e = (u, u')}\pi_e \cdot \E\left[(w_e-\wt{w}_e+\delta(R_{u'+}))^2\right] \\
   &=& \sum_{e = (u, u')}\pi_e\left((w_e-\wt{w}_e)^2 + 0 + \E\left[(\delta(R_{u'+}))^2\right]\right) \\
   &\leq& \Var[\delta_u] + \max_{e = (u, u')}\Var[\delta(R_{u'+})] \\
   &\leq& 2\e^4\E[X_{b_u}] + \max_{e = (u, u')}\Var[\delta(R_{u'+})]
\end{eqnarray*}

By Property P2, $\sum_{u \in R(v)} \E[X_{b_u}] = O(\C/\e)$ for any root-leaf path $R(v)$.
Then we have $\Var[\delta(R_\sigma)] = O(\e^3)$ by induction.
By Chebychev's inequality, we get that
$$
\sum_{v \in S \mid \left|W(v) - \wt W(v)\right| \geq \e}\Phi(v)
  = \Pr[\left|\delta(R_\sigma)\right| \geq \e]
  \leq \frac{\Var[\delta(R_\sigma)]}{\e^2}
  = O(\e).
$$

\vspace{0.3cm}
\noindent
{\bf Step 2:}
For step 2, $\left|w_e-\wt w_e\right| \leq \e w_e$ and $W(v) \leq \C$ for any $e, v$.
Thus we have that:
\begin{description}
  \item[(b)]
    $\forall v \in S, \left|W(v) - \wt W(v)\right| \leq \e\C.$
\end{description}
From (a) and (b), we can conclude the lemma.
\end{proof}

For any edge $e = (v,u)\in T(\sigma, \pi, \C)$, we have $W(u) \leq \C$ and $\wt w_e - w_e \leq \e^4$. Thus for any $v\in S_1$, we have
\begin{eqnarray*}
  \wt W(v) - W(v) \geq \max_u \left\{\wt W(u) - W(u) - \e^4 \right\}
   > (1+4\e)\C - \C - \e^4 \geq \e(\C+1).
\end{eqnarray*}
Therefore,
$
\sum_{v \in S_1}\Phi(v) = \sum_{v \in S_1 \mid \wt W(v) - W(v) > \e(\C+1)} \Phi(v).
$
Note that $S_1$ contains at most one node from each root-leaf path.
Applying the above lemma again, we have that $\sum_{v\in S_1}\Phi(v) = O(\e)$.
This completes the proof of the first part.

\vspace{0.3cm}

Now, we prove the second part. $\Phi(v)$s are defined with respect to  $T(\wt\sigma, \wt\pi, \C)$.
%
Since a canonical policy makes decisions based on the discretized sizes,
$T(\wt\sigma, \wt\pi, \C)$ has the same tree structure as $T(\wt\sigma, \pi, (1+4\e)\C)$,
except that we can not collect the profit from the later if the knapsack overflows at the end of the policy.
More precisely, we have that
$$
\P(\wt\sigma, \pi, (1+4\e)\C) = \sum_{e = (v, u) \mid W(u) \leq (1+4\e)\C} \Phi(v) p_v(w_e).
$$
where $e\in T(\wt\sigma, \wt\pi,\C)$.
W.l.o.g., we assume that $\wt W(v) \leq \C$ holds for all $v \in T(\wt\sigma, \wt\pi, \C)$.
Thus
$$
\P(\sigma, \pi, (1+4\e)\C) \geq \P(\wt\sigma, \wt\pi, \C)\biggl[1 - \sum_{v\in S_2}\Phi(v)\biggr]
$$
where
$S_2 = \{v\in T(\wt\sigma, \wt\pi, \C) \mid W(v) \leq (1+4\e)\C \text{ and } \exists e = (v, u), W(u) > (1+4\e)\C\}$.
By Lemma~\ref{lm:disctprop}, we have $\sum_{v\in S_2} \Phi(v) = O(\e)$. This completes the proof of the lemma.
\end{proof}

\noindent
{\bf Lemma~\ref{lm:ssp}}
Let $\sigma$ be the optimal policy.
Suppose $\sigma$ chooses $b$ as the next item to consider.
and decides to discard $b$ if $X_b$ realizes to $t$.
Then $\sigma$ should discard $b$ if $X_b$ realizes to a larger size $s \geq t$.

\begin{proof}
We use $\P(\sigma, \B, \C')$ to denote the expected profit the policy $\sigma$
can achieve with item set $\B$ and remaining capacity $\C'$.
Since $\sigma$ discards $b$ when $X_b$ realizes to $t$, we have that
$\P(\sigma, \B\backslash\{b\}, \C') \geq P_b + \P(\sigma, \B\backslash\{b\}, \C'-t).$
We also have $\P(\sigma, \B\backslash\{b\}, \C'-t) \geq \P(\sigma, \B\backslash\{b\}, \C'-s)$ since $\C'-t \geq \C'-s$.
Thus, $\P(\sigma, \B\backslash\{b\}, \C') \geq P_b + \P(\sigma, \B\backslash\{b\}, \C'-s)$.
Therefore, $\sigma$ should discard $b$ when $X_b = s \geq t$,
in order to maximize the expected profit.
\end{proof}


\section{\SKCRC\ Without Relaxing the Capacity}
\label{sec:norelax}
We can slightly modify the algorithm developed in~\cite{bhalgat20112} to 
give a $(2+\e)$-approximation for \SKC. 
We only state the modification.
For each item $b_i$, let $\B_i=\{b_i^t\}$ be the set of items
where $b_i^t$ corresponds to $b_i$ canceled at size $t$.
We say $b_i^t$ is a {\em small profit} item if $P_i\cdot\Pr[X_i < t]$ is less than $\e\opt$,
where $P_i$ is the (fixed) profit of $b_i$ and $X_i$ is the size of $b_i$.
Theorem 4 in~\cite{bhalgat20112} also holds for \SKC.
The only modification is how we bound the profit 
loss $\Delta_P$ in the two knapsack experiment.
Let $\mathcal{E}(b_i^t, C)$ denote the event that $b_i^t$ is placed into the first knapsack 
when the remaining capacity is $C$.
Since the items realizes to a particular size after it is placed, the loss of expected profit is
$$
\Delta_P = \sum_{b_i^t, C}\Pr\bigl[\mathcal{E}(b_i^t, C)\bigr]\cdot\Pr\bigl[C < X_i < t\bigr]\cdot P_i.
$$
Since $\Pr\bigl[C < X_i < t\bigr]= \Pr\bigl[X_i > C\bigr]\cdot\Pr\bigl[X_i < t \mid X_i > C\bigr]\leq \Pr\bigl[X_i > C\bigr]\cdot\Pr\bigl[X_i < t\bigr]$,
and $\Pr\bigl[X_i < t\bigr]\cdot P_i \leq \e\opt$ for small profit items,
we have that
$$
\Delta_P \leq \sum_{b_i^t, C}\Pr\bigl[\mathcal{E}(b_i^t, C)\bigr]\cdot\Pr\bigl[X_i > C\bigr]\cdot \e\opt.
$$
Since $\sum_{b_i^t, C}\Pr\bigl[\mathcal{E}(b_i^t, C)\bigr]\cdot\Pr\bigl[X_i > C\bigr]$ is the overflow probability, 
which is at most 1, we have $\Delta_P \leq \e\opt$.

Moreover, the $(1+\e)$-approximation algorithm (Lemma 8.1 in ~\cite{dean2008approximating})
for finding an optimal policy that packs a constant number of items can be easily extended to \SKC.
Therefore, Theorem~\ref{thm:norelax1} follows.

In \SKCRC, if we assume that, for each item, the difference between the maximum and the minimum possible sizes
is bounded by $\C-\delta$ for any constant $\delta>0$, we can still obtain a factor $2$ approximation algorithm
using the algorithm in~\cite{bhalgat20112}. 
A careful examination of the proof shows that there is no profit loss in the two knapsack experiment.
So all results in \cite{bhalgat20112} continue to hold.

\begin{thm}
\label{thm:skccnorelax}
There is a polynomial time algorithm that finds a 2-approximate
adaptive policy for \SKCRC\ 
if for each item, the difference between the maximum and the minimum possible sizes
is bounded by $\C-\delta$ for any constant $\delta>0$.
\end{thm}


\section{FPTAS for \USK}
\label{app:usk}

In this subsection, we give a linear time FPTAS for the stochastic knapsack problem
where each item has unlimited number of copies (denoted as \USK).
\USK\ is a classic Markov decision process (MDP) with continuous states
has been studied extensively \cite{derman1978renewal,derman1979renewal, assaf1982optimal,assaf1982renewal}.
Optimal adaptive policies
(in this case, the decision of inserting which item only depends on the remaining capacity)
have been characterized for some special distributions, e.g., exponential distributions~\cite{derman1978renewal}.
We would like to note that, for continuous distributions with Lipschitz PDFs, we can use the general results
on discretizing continuous MDP to get a PTAS for \USK\ (see e.g., \cite{rust1997using}).

In this problem, we first apply the discretization in Section~\ref{sec:discretization}.
In \USK, a policy can also be represented as a function $f_\sigma : [0, \C] \rightarrow \B$.
Since all discretized sizes are multiplies of $e^5$,
it is possible to find out the optimal canonical policy  $f_\sigma(\size)$ by dynamic programming as follows.

We use $\DP(\size)$ to denote that the expected profit of the optimal policy $\sigma$ on $(\wt\psi, \size)$. Initially,
$$
\DP(0) = \max_b \Bigl\{\sum_{i = 0}^{\infty}(\wt\pi_b(0))^i\cdot\wt p_b(0)\Bigr\}
= \max_b \Bigl\{\frac{\wt p_b(0)}{1 - \wt\pi_b(0)}\Bigr\},
$$
which is the expected profit of repeatedly packing copies of $b$ until one copy realizes to a nonzero size.
We compute all $\DP(\size)$ values for $\size = \size_k$ in increasing order of $k = 0, 1, \ldots, \dsize - 1$.
Suppose we choose item $b$ when the remaining capacity is $\size$.
Then $\DP(\size) = \sum_{k=0}^{\size/\e^5}[\wt p_b(\size_k) + \wt\pi_b(\size_k)\DP(\size-\size_k)]$.
Therefore, the recursion of the dynamic program is the following:
$$
\DP(\size) = \max_b \Bigl\{\Bigl(1-\wt\pi_b(\size_0)\Bigr)^{-1}\cdot
\Bigl[\wt p_b(\size_0) + \sum_{k = 1}^{\size/\e^5}\Bigl(\wt p_b(\size_k) + \wt\pi_b(\size_k)\DP(\size-\size_k)\Bigr)\Bigr]\Bigr\}.
$$
Each $\DP(\size)$ can be computed with in $O(n\cdot\C/\e^5)$ time, and there are $O(\C/\e^5)$ different values of $\size$. Thus the running time of the dynamic program is $O(n\cdot\C^2/\e^{10}) = O(n)$.
By applying the fast zero delay convolution in~\cite{dean2010speeding}, we can speed up the dynamic program and further reduce
the running time to $O(n\cdot\C/\e^{5}\log^2(\C/\e^5))$.

Since the policy can be represented as a decision tree, Lemma~\ref{lm:policytransform} still holds in this case.
Theorem~\ref{thm:usk} can be proved following the same argument as \SK.

\begin{thm}
\label{thm:usk}
For any $\e>0$, there is a linear time FPTAS for \USK\
when the knapsack capacity is relaxed to $\C+\e$.
\end{thm}

\eat{
They policy $\sigma$ is also clear when all the $\DP(\cdot)$ values are computed.
As $\sigma$ is the optimal policy on $\wt\psi$,
$$
\P(\sigma, \wt\psi, (1+3\e)\C) = (1-O(\e))\opt
$$
by result 1 of Lemma~\ref{cor:policytransform}.
We also have
$$
\P(\sigma, \psi, (1+3\e)^2\C) = (1-O(\e))\P(\sigma, \wt\psi, (1+3\e)\C) = (1-O(\e))\opt
$$
 by result 2 of Lemma~\ref{cor:policytransform}.
Thus $\sigma$ is a $1-O(\e)$ approximation policy for the original problem with the capacity relaxed by $O(\e)$.
}

\end{document}